\documentclass[letterpaper, 10 pt]{IEEEtran}  % Comment this line out
\usepackage{verbatim}
\usepackage{subcaption}
\usepackage{color}
\usepackage{doi}
\usepackage{url}
\usepackage{hyperref}
\hypersetup{
    colorlinks=true,
    linkcolor=blue,
    filecolor=magenta,      
    urlcolor=cyan,
    }
\usepackage{amsmath, amssymb}
\usepackage{braket}
\usepackage{graphicx}
\usepackage{float}
\usepackage[english]{babel}
\usepackage{amsthm}

\usepackage{tikz-layers}
\usepackage{quantikz}

\newcommand{\grad}{\text{grad\,}}
\newcommand{\Tr}{\mathrm{Tr}}
\newcommand{\argmin}{\mathrm{arg\,min\,\,}}
\newcommand{\dist}{\mathrm{dist}}
\newcommand{\e}{\mathrm{e}}
\newtheorem{definition}{Definition}
\newtheorem{theorem}{Theorem}
\newtheorem{corollary}{Corollary}
\newtheorem{lemma}{Lemma}
\title{\LARGE \bf
Noise Resilience and Robust Convergence Guarantees for the Variational Quantum Eigensolver
}

\author{Mirko Legnini and Julian Berberich% <-this % stops a space
\thanks{M. Legnini and J.Berberich are with the Institute for Systems Theory and Automatic Control,
        University of Stuttgart, Germany
        {\tt\small mirko.legnini@ist.uni-stuttgart.de}}%
\thanks{J. Berberich is also with the Center for Integrated Quantum Science and Technology (IQST), University of Stuttgart, 70569 Stuttgart, Germany}
}

\twocolumn
\begin{document}

\maketitle
\thispagestyle{empty}
\pagestyle{empty}

%%%%%%%%%%%%%%%%%%%%%%%%%%%%%%%%%%%%%%%%%%%%%%%%%%%%%%%%%%%%%%%%%%%%%%%%%%%%%%%%
\begin{abstract}
Variational Quantum Algorithms (VQAs) are a class of hybrid quantum-classical algorithms that leverage on classical optimization tools to find the optimal parameters for a parameterized quantum circuit.  One relevant application of VQAs is the Variational Quantum Eigensolver (VQE), which aims at steering the output of the quantum circuit to the ground state of a certain Hamiltonian. Recent works have provided global convergence guarantees for VQEs under suitable local surjectivity and smoothness hypotheses, but little has been done in characterizing convergence of these algorithms when the underlying quantum circuit is affected by noise.  
In this work, we derive an upper bound on the error on the optimal parameters of a VQE under the effect of different coherent and incoherent noise processes. We then procced to show robust convergence guarantees of the algorithm to the perturbed optimal parameters. Our work provides novel theoretical insight into the behavior of VQAs subject to noise. 
Furthermore, we accompany our results with numerical simulations implemented via Pennylane.    

\end{abstract}

%%%%%%%%%%%%%%%%%%%%%%%%%%%%%%%%%%%%%%%%%%%%%%%%%%%%%%%%%%%%%%%%%%%%%%%%%%%%%%%%
\section{Introduction}
Over the last few years quantum technologies have been evolving at a steadfast pace, leading to significant improvements in hardware capabilities. 
Currently available quantum computers are commonly referred to as Noisy Intermediate Scale Quantum (NISQ) devices \cite{Preskill2018}. 
The main limitations of these devices are the number of physical qubits, currently in the hundreds, and the sensitivity to noise. 
These limitations render large scale applications (e.g. \cite{Shor1997}) hard to perform, and the research community is putting significant effort into finding applications that can provide near-term quantum advantage \cite{Kim2023,Arute2019}. 
Variational Quantum Algorithms (VQAs) \cite{Cerezo2021} are one widely studied approach to this end. 
In this hybrid approach, the idea is to build a parameterized quantum circuit and employ a classical iterative optimization routine to find the set of optimal parameters that minimizes a cost function encoding the problem. 
VQAs have applications, among others, in quantum chemistry \cite{Keijzer2022,Grimsley2019}, combinatorial optimization problems \cite{Farhi2014}, and linear algebra \cite{BravoPrieto2023}.
Despite their widespread success, a main obstacle to the application of VQAs is the complexity of the underlying optimization problem.  It has been shown that finding a global minimizer for the VQA cost is, in general, an NP-hard non-convex problem \cite{Bittel2021}. 
% While there are algorithms that optimize directly over the unitary, based on the theory of Riemannian gradient flow \cite{Wiersema2023,Malvetti2024,Pervez2025}, the most widely adopted approach relies on a fixed parameterized ansatz  $U(\theta)$.
Noticeably, the structure of VQAs is very similar to that of quantum optimal control problems \cite{Magann2021,Berberich2024}. Quantum optimal control landscapes have been widely studied \cite{Chakrabarti2007}, and have already been related with the
optimization landscape of VQAs \cite{Ge2022,Huembeli2021}. However, a complete characterization of conditions for convergence to a global minimum is still lacking. Results under suitable convexity assumptions can be found in \cite{Harrow2021,Sweke2020}. 
Furthermore, the benefits of overparameterization have been highlighted via a dynamical Lie Algebra-based approach \cite{Larocca2023}. Recently, \cite{Wiedmann2025} proposed global convergence guarantees for the Variational Quantum Eigensolver (VQE) based on local surjectivity.
The VQE is a commonly utilized VQA \cite{Peruzzo2014,Tilly2022}, where the cost function is chosen as a quantum measurement with respect to some observable $O$, and the goal of the parameterized circuit is to steer the system to its ground state.
\newline

One reason why VQAs are a good fit to NISQ hardware is their inherent robustness to coherent errors. In particular, it can be shown that, for systematic gate biases, 
the error introduced by the biases can be automatically corrected by the algorithm by converging to a different parameter vector, as long as the error respects the circuit symmetries \cite{McClean2016}.
Robustness of VQAs against specific rotation errors and its links to generalizability in quantum machine learning has been studied in \cite{Berberich2024a,Berberich2024b}. 
A different angle to look at robustness of VQAs is understanding how the optimal parameters change due to noise. 
For some applications, e.g. Variational Quantum Compiling (VQC) \cite{Khatri2019}, the output of the algorithm is the set of parameters 
to which we converge rather than the optimal cost or final state. Robustness of VQC was studied in \cite{Sharma2020}, where they found that 
optimal parameters remain unchanged for specific error classes. 
The question of how general errors affect the optimal parameters of a given VQE remains largely unanswered. 
\newline

The contribution of this work is twofold. First, we derive upper bounds for the distance between the optimal parameters of a noise-freeVQE and the corresponding noise-perturbed case. We define this upper bound as the noise resilience of the VQE.
In order to derive upper bounds, we provide a notion of perturbation level for both coherent and incoherent noise processes. We then show that the 
optimal parameter distance grows at most polynomially with the perturbation level.  
We further show that, for classes of well-behaved cost functions, the error bound grows linearly with the perturbation level. 
Our approach combines tools from optimization theory, differential geometry and Morse theory, and is the first to provide explicit error bounds on general noise models and variational algorithms.   
Second, we study robust convergence of locally surjective VQEs, showing that for a small enough perturbation level the algorithm converges asymptotically to a neighborhood of the noise-free global minimum.
\\

The rest of the paper is structured as follows.
In Section II, we provide an introduction to VQAs and VQE convergence guarantees and on quantum noise. 
In Section III, we introduce the perturbed observable VQE problem and we show noise resilience by deriving upper bounds on 
the distance between any optimal parameter of the perturbed problem and the optimal parameter set of the 
noise-free problem.
In Section IV, we show robust convergence guarantees for the perturbed observable problem under local surjectivity and small noise assumptions. 
In Section V, we show that VQEs affected by noise can be reformulated as noise-free VQEs with perturbed observables, such that the results from the previous two sections apply. 
In Section VI, we discuss specific cases where stronger statements can be made. 
Finally, in Section VII, we show and discuss numerical results that support our theory.

%%%%%%%%%%%%%%%%%%%%%%%%%%%%%%%%%%%%%%%%%%%%%%%%%%%%%%%%%%%%%%%%%%%%%%%%%%%%%%%%
\section{Preliminaries}
In this section, we introduce the basic framework of VQAs (Section II.A).
We also discuss the formalism we use to model coherent and incoherent errors acting on the circuit (Section II.B). 

Throughout this paper, $n$ indicates the number of qubits in the circuit and $N=2^n$ the dimension of the Hilbert space. $\mathbb{H}^{N \times N}$ indicates the set of $N \times N$ Hermitian matrices, that is $H \in \mathbb{C}^{N \times N} : H=H^\dagger$. $SU(N)$ indicates the special unitary group of dimension $N$ and $\mathfrak{su}(N)$ indicates the special unitary algebra.
We denote by $\braket{A,B}=\Tr[A^\dagger B]$ the Hilbert-Schmidt inner product and by $[A,B]=AB-BA$ the commutator of $A$ and $B$.
We denote the distance between a set $Y$ and a point $x$ by $\dist(x, Y)= \min_{y \in Y} \|y-x\|$. 
Given a vector $v$, the infinity norm $\|v\|_\infty$ is defined as $\|v\|_\infty=\max_i |v_i|$.
Given a smooth function $f: X \rightarrow \mathbb{R}$ with a compact domain $X$, $f$ is uniformly upper bounded on $X$ and its derivatives up to any order are Lipschitz continuous. 
With $L_f^{(p)}$ we indicate the Lipschitz constant of the $p$-th derivative of $f$. With $M_f$ we indicate the upper bound.  

\subsection{Variational Quantum Algorithms}

In VQAs, we consider optimization problems of the form 
\begin{align}
   \min_{\mathbf{\theta} \in \Theta} \ell (\mathbf{\theta}) = \min_{\mathbf{\theta} \in \Theta} \Tr[O U(\mathbf{\theta})\rho_0 U^\dagger(\mathbf{\theta})] \label{nominalproblem},
\end{align}
where $\rho_0 \in \mathbb{H}^{N \times N}$ is an easy-to-prepare $N$-dimensional input state, $O \in \mathbb{H}^{N \times N}$ an observable.
Further, the parameter space $\Theta$ can be, for example, $\mathbb{R}^p$, or the $p$-dimensional torus $\mathbb{T}^p$.
The unitary map $U: \Theta \rightarrow SU(N)$ has the form 
\begin{align}
   U(\theta)=\prod_{j=1}^L U_j(\theta_j), \label{qcirc}
\end{align}
where each $U_j(\theta_j)$ is a unitary operator.
For compactness of representation, we define $\rho(\mathbf{\theta})= U(\mathbf{\theta})\rho_0 U^\dagger(\mathbf{\theta})$. 
Note that, for pure states, $\rho(\theta)$ can also be written as $\ket{\psi(\theta)}\bra{\psi(\theta)}$, where $\ket{\psi(\theta)} \in \mathbb{C}^N$ is a vector and $\bra{\psi(\theta)}=\ket{\psi(\theta)}^\dagger$. 
We call $\Theta^* \subset \Theta$ the set of globally optimal parameters for the problem \eqref{nominalproblem}. We also call $\Theta^\mathcal{C}$ the set of critical points of $\ell$, i.e. the points such that $\nabla \ell(\theta)=0$.

Typically, the unitaries are represented as 
\begin{align}
   & U_j(\mathbf{\theta}_j)=\prod_{k=1}^{p_j} e^{-i\theta_{j,k}H_{j,k}}
\end{align}
with $H_{j,k} \in \mathbb{H} ^{N \times N}$ for $k=1,..., p_j $ being the set of Hermitian generators for the $j$-th layer. Examples of architectures that satisfy this structure are QAOA \cite{Farhi2014}, the hardware-efficient ansatz \cite{Kandala2017} and the variational Hamiltonian ansatz \cite{Wiersema2020}.
Another possible structure for the layers is
\begin{align}
   & U_j(\mathbf{\theta}_j)=e^{-i A_j(\mathbf{\theta_j}) } \\
   & A_j(\mathbf{\theta}_j)=\sum_{k=1}^{p_j} \theta_{j,k} H_{j,k}.   
\end{align}
This is the $SU(N)$ ansatz as
for $p_j \geq N$ and suitable choice of $\{H_{j,k} | k=1,\dots, p_j\}$ it generates the special unitary group of dimension $N$ \cite{Wiersema2024}.

Throughout this paper, we assume without loss of generality that $\min_{\mathbf{\theta} \in \Theta} \ell (\mathbf{\theta})=0$. 
In case the minimum takes a different value $\ell^*$, the cost function can be shifted as $\ell'(\mathbf{\theta})=\ell(\theta)-\ell^*$ for our analysis. 

In order to find the minimum, we consider a gradient descent algorithm of the form
\begin{align}
   \theta^{k+1}=\theta^{k}-\gamma \nabla \ell(\theta^k), \label{gradientdescent}
\end{align}
where $\gamma>0$ is small enough to guarantee a decrease in the cost function. If $\ell$ is $L$-smooth (i.e. $\nabla \ell$ is $L$-Lipschitz), then $\gamma \le \frac{2}{L}$ satisfies this condition \cite{simonetto2025}. 
Notice that due to the non-convex nature of $\ell$ in general only convergence to a local minimum can be guaranteed.

In order to study convergence to a global minimum, it is convenient to consider the cost function as the composition of the map $U: \Theta \rightarrow SU(d)$ and a cost functional $L: SU(d) \rightarrow \mathbb{R}$.
In view of this, the partial derivatives of \eqref{nominalproblem} can be represented as \cite{Wiedmann2025}

\begin{align}
   \dfrac{\partial}{\partial \theta_j} \ell(\theta)=\braket{\grad L [U (\mathbf{\theta})], \Omega_j (\mathbf{\theta})},
\end{align}
with 
\begin{align}
   & \grad L[U(\theta)]=[O, \rho(\mathbf{\theta})] \text{, and}\\
   & \Omega_j(\mathbf{\theta})= U^\dagger(\mathbf{\theta}) \dfrac{\partial}{\partial \theta_j} U(\mathbf{\theta}).
\end{align}
Here, $\grad L[U(\mathbf{\theta})]$ is the Riemannian gradient of the cost functional $L$ and $\Omega_j(\mathbf{\theta})$ the rate of change of the unitary with respect to the variation of the parameter. 
Both vectors are elements of the tangent space at the identity of $SU(N)$.

In this work, we will focus on locally surjective parameterized unitaries. 
\begin{definition}
   A parameterized unitary transformation $U$ is said to be locally surjective if 
   \begin{align}
      \mathrm{span}\{\Omega_j(\mathbf{\theta}) | j=1,\dots, p\}=\mathfrak{su}(N),  \forall \mathbf{\theta} \in \Theta.
   \end{align} 
\end{definition}

In other words, local surjectivity guarantees that it is always possible to move along any direction of the special unitary group manifold by varying some parameter $\theta_j$. 
An important consequence of a locally surjective parametrized unitary is that one can derive global convergence guarantees.
We state here the main theorem from \cite{Wiedmann2025}.   
\begin{theorem} 
   Let $U$ be a locally surjective parametrized unitary. Then, for almost all initial conditions $\theta_0$, using gradient descent $\eqref{gradientdescent}$, the parameter $\theta^k$ either converges to a global minimum of $\ell(\theta)$ or diverges as $k\rightarrow\infty$. 
\end{theorem}
A proof for this theorem as well as examples of locally surjective unitary transformations can be found in \cite{Wiedmann2025}.
Assuming additionally compactness of the parameter space convergence to a global minimum can be guaranteed. It has to be noted, however, that the construction of a locally surjective and compact ansatz is still an open problem. 
\\

It is important to remark that the solution to the optimization problem \eqref{nominalproblem} is in general non-unique. This could be both due to a degenerate ground state of the observable $O$
or non-injectivity of the Ansatz.
To characterize optimal solutions for VQEs, we will later focus our attention on cases in which the cost $\ell$ is a Morse-Bott function, which is defined as follows \cite[Def 2.6.1]{Nicolaescu2011}.
\begin{definition}
   Let $\mathcal{M}$ be a smooth manifold and $f: \mathcal{M} \rightarrow \mathbb{R}$ be a smooth function. Then $f$ is said to be
   Morse-Bott if:
   \begin{itemize}
      \item[1.] The set of critical points of $f$ is a disjoint union of connected, smooth submanifolds $\mathcal{M}_k$ such that
      $f$ is constant on each component, and
      \item[2.] for each $p \in \mathcal{M}_k$, the kernel of the Hessian of $f$ at $p$ is exactly the tangent space $T_p \mathcal{M}_k$.
   \end{itemize}
\end{definition}
This can be understood as a function whose critical points are found in manifolds and whose local behavior is quadratic in directions 
transverse to the critical set itself. In this sense, it can be interpreted as a set-valued generalization of isolated critical points with nondegenerate Hessian, allowing flat
directions tangential to the critical manifold. 

\subsection{Quantum Noise}
In the following, we introduce the noise sources we deal with in this work. 
We start by defining coherent errors. 

Given an ideal quantum circuit \eqref{qcirc},  a coherent error $U_{\e,j}=e^{-i H_{\e,j}}$ is a unitary operator such that 
the resulting, noise-affected quantum circuit reads
\begin{align}
   \tilde{U}(\theta)=\prod_j^L \tilde{U_j}(\theta_j)= \prod_i^L U_j(\theta_j) U_{\e,j}.
\end{align}
The generators $H_{\e,j}$ are in general not known, and can be different for each layer $j$. It is however assumed that they are elements of a known set $\mathcal{H}_{\e, j}$.
For example, $U_{\e,j}$ could be a single-qubit Pauli rotation with unknown but bounded angle $U_{\e,j} \in \{e^{-i\eta X} | \eta \in [\underline\eta, \bar \eta]\}$.

We proceed by introducing incoherent quantum errors. 
An ideal quantum circuit \eqref{qcirc} affected by incoherent errors is described by means of the quantum operation
\begin{align}
   \mathcal{E}=\mathcal{E}_{L} \circ \mathcal{E}_{L-1} \circ \dots \circ \mathcal{E}_{1}, 
\end{align}
such that each layer acts on a density operator $\rho$ as  
\begin{align}
   \mathcal{E}_{j}(\rho)=\sum_{k=1}^{m} E_k U_j(\theta_j) \rho U^\dagger_j(\theta_j) E^\dagger_k,
\end{align}
for some set of evolution operators $E_k$. Quantum operations can model, for example, Pauli channels, depolarization and amplitude damping. An overview of incoherent errors can be found in \cite[Chapter 8]{Nielsen2012}.     
\\

%%%%%%%%%%%%%%%%%%%%%%%%%%%%%%%%%%%%%%%%%%%%%%%%%%%%%%%%%%%%%%%%%%%%%%%%%%%%%%%%%%%%%%%%%%%%%%%%%%%%%%%%%%%%%%%%%%%%%%%%%%%%%%%%%%%%%%%%%%%%%%%%%%%%%%%%%%%%%%%%%%%%%%%%%%%%%%%%%%%%%%%%%%%%%%%%%%%%%%%%%%%%
\section{Noise Resilience of Optimal Parameters}
In the following, we consider the perturbed VQE 
\begin{align}
   \min_{\theta \in \Theta} \tilde{\ell} (\theta) = \min_{\theta \in \Theta} \Tr[(O+\epsilon \tilde{O}(\theta)) U(\mathbf{\theta}) \rho_0 U(\mathbf{\theta})^\dagger] \label{perturbedobs}
\end{align}
The main difference to the noise-free VQE \eqref{nominalproblem} is the perturbation of the observable via the term $\epsilon\tilde{O}(\theta)$.
In Sections III and IV, we derive noise resilience and robust convergence guarantees for such perturbed observable VQEs \eqref{perturbedobs}.  
In Section V, we will derive explicit expressions for $\tilde{O}(\theta)$ for common noise classes, including both coherent and incoherent errors. 
As we will see, the observable will in general depend on the parameters $\theta$. 
We assume in the following that this dependence is smooth, which is fulfilled for all noise classes considered in this paper (see Section V for details). 
The perturbation level $\epsilon >0$ is a quantification of the noise intensity. 
We indicate with $\tilde{\Theta}_\epsilon^*$ the set of optimal solutions of \eqref{perturbedobs} for a given $\epsilon$. %, and with $\tilde{\theta}_\epsilon^*$ an arbitrary element of $\tilde{\Theta}_\epsilon^*$ .

\subsection{General Upper Bound}
In this setup, we can formulate our main result.
\begin{theorem}
   Suppose $\Theta$ is compact.
   Then, there exist constants $C, \bar\epsilon>0$ and $\alpha \in (0,1)$ such that
\begin{align}
   \mathrm{dist}({\tilde{\theta}^*_\epsilon, \Theta^*})\le C \epsilon^{\frac{1-\alpha}{\alpha}},\label{ebound}
\end{align} 
for any  $\epsilon < \bar\epsilon$ and $\tilde{\theta}_\epsilon^* \in \tilde{\Theta}_\epsilon^*$.
\end{theorem}

\begin{proof}
The proof is structured in two steps.   
The first step of the proof is showing that $\tilde\theta^*_\epsilon$ is in a neighborhood of $\Theta^*$ for $\epsilon$ sufficiently small. 
The second step provides the explicit bound \eqref{ebound}. 

\textbf{Step 1 - Continuity bound:}
We start by defining a sequence $\{\epsilon_k \}_k$ such that $\epsilon_k \rightarrow 0$. 
We then define the set of optimal parameters $\tilde{\Theta}_{\epsilon_k}^*$ for the perturbed cost $\tilde{\ell}_k(\theta)=\Tr[(O+\epsilon_k \tilde{O}(\theta)) U(\mathbf{\theta}) \rho_0 U(\mathbf{\theta})^\dagger] $ 
for a specific element of the sequence $\epsilon_k$ via
\begin{align}
   \tilde{\Theta}_{\epsilon_k}^* = \argmin \ell(\theta) + \epsilon_k g(\theta),
\end{align}
with $g(\theta) = \Tr[\tilde{O}(\theta) U(\theta) \rho_0 U^\dagger(\theta)]$. 

It is immediate to see that $\Theta^*_0=\Theta^*$. Let $\tilde{\theta}_{\epsilon_k}^* \in \tilde{\Theta}_{\epsilon_k}^*$ be an arbitrary perturbed optimal parameters and $\theta^* \in \Theta^*_0$  be an arbitrary optimal parameter for the noise-free cost $\ell(\theta)$.
By optimality, we have 
\begin{align}
   \ell(\tilde{\theta}_{\epsilon_k}^*) + \epsilon_k g(\tilde{\theta}_{\epsilon_k}^*) \le \ell(\theta^*) + \epsilon_k g(\theta^*).
\end{align}
By observing that $\ell(\theta^*)=0$ and rearranging the terms, we get
\begin{align}
   & \ell(\tilde{\theta}_{\epsilon_k}^*) \le \epsilon_k(g(\theta^*) - g(\tilde{\theta}_{\epsilon_k}^*)) \nonumber \\
   &\le \epsilon_k (|g(\theta^*)| + |- g(\tilde{\theta}_{\epsilon_k}^*)|) \nonumber \\
   %&\le \epsilon_k (\|\tilde{O}(\theta^*)\| + \|\tilde{O}(\theta^*)\|) \nonumber \\
   &\le 2 \epsilon_k M_g,
\end{align}
where we used boundedness of $g$ on $\Theta$.
%where we used the fact that $\Tr[\tilde{O}(\theta) U(\theta) \rho_0 U^\dagger(\theta)] \le \|\tilde{O}(\theta)\|$.

Taking the limit, we get 
\begin{align}
   \lim_{k \rightarrow \infty} \ell(\tilde{\theta}_{\epsilon_k}^*) \le \lim_{k \rightarrow \infty} 2 \epsilon_k M_g = 0. 
\end{align}
Compactness of $\Theta$ guarantees the existence of at least one accumulation point $\bar\theta \in \Theta$ for the sequence $\{\tilde{\theta}^*_{\epsilon_k}\}_k$. 
With this fact and by continuity of $\ell$ we can state that $\ell(\bar\theta) = 0$, and therefore $\bar\theta \in \Theta^*_0$. 
This shows that $\tilde{\theta}^*_\epsilon$ approaches the set $\Theta^*_0$ asymptotically as $\epsilon$ decreases to $0$, meaning that, for $\epsilon$ small enough, $\tilde{\theta}_{\epsilon}^*$ lies arbitrarily close to $\Theta^*$.

\textbf{Step 2 - Explicit bound scaling:} In this step, we derive the explicit bound \eqref{ebound}. 
By the first order optimality condition, we have 
\begin{align}
   \nabla\ell(\tilde{\theta}^*_\epsilon)=-\epsilon \nabla g(\tilde{\theta}^*_\epsilon). \label{nablas} 
\end{align}
Since $\ell$ is a real analytic function, we can now use the Kurdyka-\L{}ojasiewicz inequalities.
In particular \cite[Corollary 2]{Kurdyka2014} states that there exist constants $K, r_0>0$ and $\alpha \in (0,1)$ such that
\begin{align}
   \| \nabla \ell(\tilde{\theta}^*)\|  \ge {K} \dist({\tilde{\theta}^*_\epsilon, \Theta^*})^{\frac{\alpha}{1-\alpha}} \label{KL} 
\end{align}
holds if $\dist({\tilde{\theta}^*_\epsilon, \Theta^*})<r_0$, that is, if $\tilde{\theta}^*_\epsilon$ is in a neighborhood of $\Theta^*$.
By the first step of the proof, for each $r>0$ we can always find $\hat\epsilon(r)$ such that $\dist(\tilde{\theta}_{\hat\epsilon(r)}^*, \Theta^*) < r$.  
We define $\bar{\epsilon}=\hat\epsilon(r_0)$.
If $\epsilon < \bar\epsilon$, we can manipulate \eqref{KL} to get
\begin{align}
   \dist({\tilde{\theta}^*_\epsilon, \Theta^*}) \le (\frac{1}{K} \| \nabla \ell(\tilde{\theta}^*)\|^{\frac{1-\alpha}{\alpha}}). \label{LK} 
\end{align}

Then, substituting \eqref{nablas} in \eqref{LK}, we obtain
\begin{align}
   & \dist({\tilde{\theta}^*_\epsilon, \Theta^*}) \le (\frac{1}{K} \| -\epsilon \nabla g(\tilde{\theta}^*)\|)^{\frac{1-\alpha}{\alpha}} \nonumber \\
   & = (\frac{1}{K} \| \nabla g(\tilde{\theta}^*)\|)^{\frac{1-\alpha}{\alpha}}\epsilon^{\frac{1-\alpha}{\alpha}} \nonumber \\
   & \le (\frac{1}{K} (M_{\| \nabla g\|}))^{\frac{1-\alpha}{\alpha}}\epsilon^{\frac{1-\alpha}{\alpha}} \nonumber \\
   & = C \epsilon^{\frac{1-\alpha}{\alpha}}, \label{bound}
\end{align}
%where we used the facts that  $\|\nabla g (\theta)\| \le 2 \sqrt{L} \|H_{\max}\| \, \| \tilde{O}(\theta) \|~+~ M_2,$   and  $\|\tilde{O}(\theta)\| \le M_1$.
%By $\|H_{\max}\|$ we indicate the maximum Hamiltonian norm $\|H_{\max}\|= \max_j \|H_j\|$. 
This proves the theorem statement. 
\\

\end{proof}
\begin{figure}
   \includegraphics[width=\columnwidth]{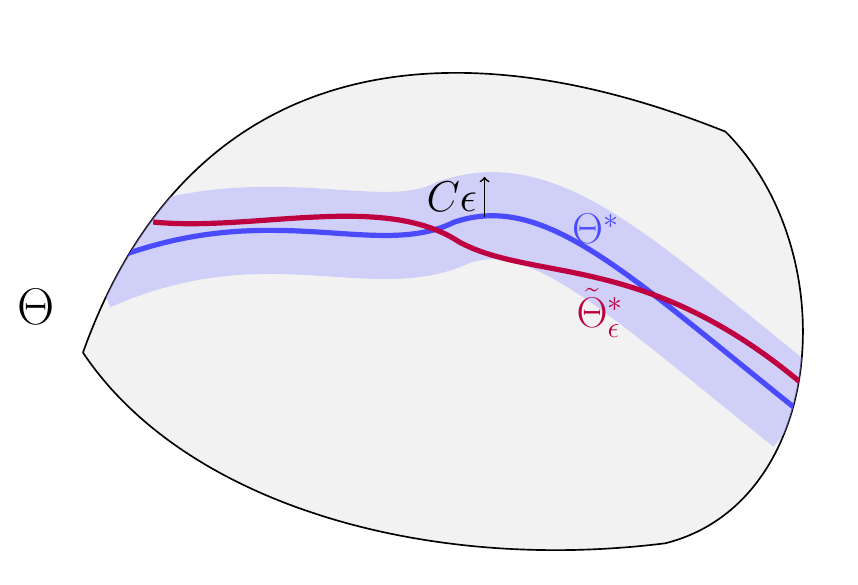}
   \caption{Visual representation of the set of unperturbed optimal parameters $\Theta^*$ and the set $\tilde{\Theta}^*_\epsilon$. The shaded area around $\Theta^*$ indicates the region in which $\tilde{\Theta}^*_\epsilon$ can lie according to our upper bound.}
   \label{sets}
\end{figure}

In Theorem 2 we derive a worst-case upper bound for the distance $\dist({\tilde{\theta}^*_\epsilon, \Theta^*})$ between solutions of \eqref{nominalproblem} and solutions of \eqref{perturbedobs} for small perturbations. 
In particular we show show that, for small perturbations, we can derive a polynomial scaling 
$\mathrm{dist}({\tilde{\theta}^*_\epsilon, \Theta^*})\le C \epsilon^{\frac{1-\alpha}{\alpha}}$. 

We prove the first part by providing an upper bound on the variation of the cost function $\ell(\theta^*_\epsilon)$,
and using a continuity argument and compactness of the parameter set 
to show that the perturbed optimal parameter lies in a neighborhood of the noise-free optimal parameter set.
The situation is illustrated in Figure~\ref{sets}.
%Note that the assumption that $\|\tilde{O}(\theta)\|$ and $\|\nabla\tilde{O}(\theta)\|$ are bounded holds for arbitrary VQEs affected by general coherent errors or incoherent errors (the latter subject to conditions discussed in Section III.C).

The proof relies on \L{}ojasiewicz inequalities, which are used to characterize the behavior of analytic functions close to critical sets 
and are widely employed in optimization theory \cite{Colding2014}.
It is important to notice that such inequalities allow one to make local statements around minima only, characterizing how they shift in the parameter space due to the noise. The upper bound on the noise $\bar\epsilon$ is exactly due to this: while the \L{}ojasiewicz inequalities
allow us to characterize how the optimization landscape is perturbed locally around critical points, we need noise to be small enough that no new critical points appear far from the noise-free ones. The existence of the noise upper bound $\bar\epsilon$ is indeed guaranteed by compactness of the parameter space, 
which implies that the minimum norm of the gradient has a lower bound far from critical points. 
The \L{}ojasiewicz exponent $\alpha$ in \eqref{KL} gives information on the conditioning of a generic (possibly non-convex) optimization problem near the set of minimizers. In particular, it is related to the exponent of the polynomial growth of the function around its minimum. 
In the nondegenerate case (the cost function having an isolated minimum with a positive definite Hessian or satisfying a quadratic growth condition around optimal values) the landscape is locally strongly convex with $\alpha=\frac{1}{2}$.
Theorem~2 shows that, in this case perturbations of the cost induce proportional perturbations on the optimal parameters.
On the other hand, if the optimization landscape is flat around the optimal parameter set $\theta^*$, we have that $\alpha >\frac{1}{2}$ \cite{Feehan2020}. In this case, it follows that ${\frac{1-\alpha}{\alpha}} <1$, 
which implies that the slope of $\epsilon^{\frac{1-\alpha}{\alpha}}$ at 0 increases. Therefore, the sensitivity of errors in the parameter space to perturbations is higher. 
For example, problem-inspired architectures that rely on symmetries in the Hamiltonian to build the parametrized circuit can admit higher noise sensitivity due to a poor conditioning of the cost function around the solution.   

\subsection{Conditions for Linear Scaling}
We proceed by characterizing cases that guarantee a linear scaling in $\epsilon$ of the worst-case parameter error. 
This property turns out to be related to quadratic growth around the optimal parameter set $\Theta^*$, which is the guarantee that there exist $\mu, r>0$ such that $\ell(\theta) \ge \mu \, \dist(\theta, \Theta^*)^2$ for $\dist(\theta, \Theta^*) < r$. 
In order to do so, we introduce the quantum geometric tensor $F(\theta)$ defined via 
\begin{equation}
   F_{i,j}(\theta)=\bra{\frac{\partial}{\partial\theta_i}\psi(\theta)} I-\rho(\theta) \ket{\frac{\partial}{\partial\theta_j} \psi(\theta)},
\end{equation}
for $i,j \in \{1,2,...,p\}$. The quantum geometric tensor is the real part of the quantum Fischer information matrix, and it can be understood as the Fubini-Study metric on the density operator manifold \cite{cheng2013}. 
Roughly speaking, it indicates the sensitivity of the density operator to changes in the parameters. A zero eigenvalue in $F(\theta)$ implies a constant direction in the parameter space.
The following corollary provides an account of cases in which we can guarantee quadratic growth around $\Theta^*$, thereby guaranteeing linear scaling in the error upper bound. 
\begin{corollary}
   Suppose $\Theta$ is compact.
   Suppose additionally that at least one of the following two conditions holds: 
   \begin{itemize}
      \item $U$ is locally surjective, or 
      \item the observable $O$ is nondegenerate and the quantum geometric tensor has full rank for each $\theta \in \Theta$.
   \end{itemize} 
   Then there exist constants $C, \bar\epsilon>0$ such that $\mathrm{dist}({\tilde{\theta}^*_\epsilon, \Theta^*})\le C \epsilon,$ 
for any  $\epsilon < \bar\epsilon$ and $\tilde{\theta}_\epsilon^* \in \tilde{\Theta}_\epsilon^*$.
\end{corollary} 

\begin{proof}
   Assuming local surjectivity, \cite[Appendix C]{Wiedmann2025} proves that $\ell (\theta)$ is a Morse-Bott function. This guarantees the exponent $\alpha=\frac{1}{2}$ \cite{Feehan2020}. 

Assuming full rank of the quantum geometric tensor and nondegeneracy of $O$, \cite{Zunkovic2026} shows directly that the KL condition \eqref{KL} is satisfied with $\alpha=\frac{1}{2}$.
In both cases the result then follows from \eqref{bound}.  

\end{proof}
The linear scaling $C \epsilon$ is the best-behaved we can have, as for small errors it leads to the smallest possible sensitivity.
The local surjectivity condition can only be satisfied with a number of parameter $p\ge\text{dim}(SU(N))=4^n-1$, while the quantum geometric tensor rank condition can only be satisfied for underparametrized unitaries $p \le N=2^n$. 
Interestingly, we have an asymmetry in the conditions stated in the corollary: 
while we impose a local surjectivity condition on $U$, which in turn implies local surjectivity of $\rho(\theta)$, 
the condition on the quantum geometric tensor requires directly that $\rho(\theta)$ is locally injective, with no general condition on the parametrized unitary map $U$. 
The question remains open whether it is possible to guarantee linear scaling on the error assuming only local surjectivity of $\rho(\theta)$ without surjectivity of the unitary.   
This would fill the gap for linear scaling of the parameter error in the perturbation level for a number of parameters $2^n < p <4^n-1$. 
\\

\section{Robust convergence guarantees}
In Section III, we have shown that the optimal parameters of the perturbed problem \eqref{perturbedobs} differ from the noise-free optimal parameters of \eqref{nominalproblem} with an upper bound that increases with the perturbation level.
Next, we investigate asymptotic convergence of the perturbed observable VQE \eqref{perturbedobs}.
In particular, we show that such a VQE with locally surjective parametrized unitary on a compact parameter set converges asymptotically to its global minimum if the perturbation level is small enough.
 
\begin{theorem}
   Suppose $\Theta$ is compact and $U$ is locally surjective.
   Then, there exists $\bar \epsilon>0$ such that, if $\epsilon < \bar \epsilon$, using gradient descent $\eqref{gradientdescent}$, the parameter $\theta^k$ converges to a global minimum of $\tilde\ell(\theta)$ as $k\rightarrow\infty$, for almost all initial conditions $\theta_0$. 
\end{theorem}

\begin{proof}
   The main idea of the proof is to show that, for $\epsilon$ small enough, the index of each singular point of \eqref{perturbedobs} remains the same as the noise-free problem \eqref{nominalproblem}, showing in particular that no saddles of the noise-free problem turn into local minima because of the perturbation. 
   The result then follows from noise-free convergence results given in Theorem 1 and compactness of $\Theta$. 
   In order show that saddles are preserved, we first bound the distance between the noise-free critical points and the perturbed ones. This result is a generalization of the upper bound in the previous section to general critical sets. Then we use this distance bound to bound the eigenvalue variation of the Hessians. 
   
   \textbf{Step 1 - Parameter variation} 
   First, we define the set of perturbed critical points $\Theta^\mathcal{C}_{\epsilon}$ as the set  $\{\theta \in \Theta | \nabla \tilde{\ell}(\theta)=0\}$, for $\tilde \ell$ defined as in \eqref{perturbedobs}.
   Next, we focus our attention on a perturbed critical parameter $\theta^\mathcal{C}_\epsilon \in \Theta^\mathcal{C}_{\epsilon}$ and its projection on the noise-free critical set $\theta^\mathcal{C} \in \Theta^\mathcal{C}$.
   By definition, for this pair $\dist(\theta^\mathcal{C}_\epsilon, \Theta^\mathcal{C})=\|\theta^\mathcal{C} - \theta^\mathcal{C}_\epsilon\|$.
   To bound this distance, consider \eqref{nablas}. Taking a first-order expansion of the left-hand term yields
   \begin{align}
      \nabla \ell(\theta^\mathcal{C}_\epsilon)=\nabla^2 \ell(\theta^\mathcal{C})(\theta^\mathcal{C}_\epsilon-\theta^\mathcal{C}) + O(\|\theta^\mathcal{C} - \theta^\mathcal{C}_\epsilon\|^2),
   \end{align}  
   since the zeroth-order term $\nabla \ell (\theta^\mathcal{C})=0$. 
   Since all critical manifolds considered are either global extrema or strict saddles, and $\theta^\mathcal{C} - \theta^\mathcal{C}_\epsilon$ only has components on the 
   transverse subspace, 
   \begin{align}
      & \|\nabla\ell(\theta^\mathcal{C}_\epsilon)\| = \|\nabla^2 \ell(\theta^\mathcal{C})(\theta^\mathcal{C}_\epsilon-\theta^\mathcal{C}) + O(\|\theta^\mathcal{C} - \theta^\mathcal{C}_\epsilon\|^2)\| \\   
      & \ge \mu\|\theta^\mathcal{C} - \theta^\mathcal{C}_\epsilon\| - O(\|\theta^\mathcal{C} - \theta^\mathcal{C}_\epsilon\|^2), \label{perpnabla}
   \end{align}  
   
   where $\mu \in \sigma(\nabla^2 \ell (\theta^\mathcal{C}))$ is the smallest non-zero singular value of the Hessian evaluated at $\theta^\mathcal{C}$. 
   Substituting \eqref{perpnabla} in \eqref{nablas} we get 

   \begin{align}
      & \|\nabla g(\theta^\mathcal{C}_\epsilon)\| \epsilon \ge \mu\|\theta^\mathcal{C} - \theta^\mathcal{C}_\epsilon\| - O(\|\theta^\mathcal{C} - \theta^\mathcal{C}_\epsilon\|^2)\\
      & \ge \mu\|\theta^\mathcal{C} - \theta^\mathcal{C}_\epsilon\| -   \frac{L_\ell^{(1)}}{2}\|\theta^\mathcal{C} - \theta^\mathcal{C}_\epsilon\|^2,
      \end{align}
   where we used $L$-smoothness of $\ell$ to upper bound the higher-order terms \cite[Lemma 1.2.3]{Nesterov2004}.
   Rearranging and defining $\delta:=\|\theta^\mathcal{C} - \theta^\mathcal{C}_\epsilon\|$, we get the condition
   \begin{align}
      & \frac{L_\ell^{(1)}}{2}\delta^2  -\mu \delta + \|\nabla g(\theta^\mathcal{C}_\epsilon)\| \epsilon \ge 0. \label{quadtheta}   
   \end{align}
   The roots of the polynomial on the left-hand side of \eqref{quadtheta} are located in $\delta_{1,2}= \frac{\mu \pm \sqrt{\mu^2 - 2 L_\ell^{(1)} \|\nabla g(\theta^\mathcal{C}_\epsilon)\|\epsilon }}{L_\ell^{(1)}}$. 
   % = \frac{2\|\nabla g(\theta^\mathcal{C}_\epsilon)\|\epsilon}{\mu \mp \sqrt{\mu^2 - 2 L_\ell^{(1)} \|\nabla g(\theta^\mathcal{C}_\epsilon)\|\epsilon} }$.
   The roots are real if 
   \begin{align}
      \mu^2 - 2 L_\ell^{(1)} \|\nabla g(\theta^\mathcal{C}_\epsilon)\|\epsilon \ge 0,
   \end{align}
   which is satisfied for $\epsilon \le \frac{\mu^2}{2L M_{\|\nabla g\|}}=:{\epsilon}_\Delta$.
   The inequality \eqref{quadtheta} is then satisfied for 
   \begin{align}
   \delta \in [0,\delta_1] \cup [\delta_2, +\infty). \label{range}
   \end{align} 
   The geometric interpretation of this union of two intervals is that the $L$-smoothness of the cost forces singular points of \eqref{perturbedobs} to be either shifted versions of the noise-free ones lying in a neighborhood or new points generated far enough from the noise-free ones. 
   
   We focus now on the left branch of solutions $[0, \delta_1]$. 
   The admissible range of solutions in this branch is 
   \begin{align}
      &\|\theta^\mathcal{C} - \theta^\mathcal{C}_\epsilon\|= \delta \le \delta_1 = \frac{\mu - \sqrt{\mu^2 - 2 L_\ell^{(1)} \|\nabla g(\theta^\mathcal{C}_\epsilon)\|\epsilon }}{L_\ell^{(1)}} \\
      &=  \frac{\mu - \sqrt{\mu^2 - 2 L_\ell^{(1)} \|\nabla g(\theta^\mathcal{C}_\epsilon)\|\epsilon }}{L_\ell^{(1)}}
      \frac{{\mu + \sqrt{\mu^2 - 2 L_\ell^{(1)} \|\nabla g(\theta^\mathcal{C}_\epsilon)\|\epsilon }}}{{\mu + \sqrt{\mu^2 - 2 L_\ell^{(1)} \|\nabla g(\theta^\mathcal{C}_\epsilon)\|\epsilon}}} \\
      &= \frac{2\|\nabla g(\theta^\mathcal{C}_\epsilon)\|\epsilon}{{\mu + \sqrt{\mu^2 - 2 L_\ell^{(1)} \|\nabla g(\theta^\mathcal{C}_\epsilon)\|\epsilon}}}
      \le \frac{2\|\nabla g(\theta^\mathcal{C}_\epsilon)\|\epsilon}{{\mu}}\\ &\le \frac{M_{\|\nabla g\|}}{{\mu}}\epsilon :=L_\theta \epsilon,  \label{distancebound}
   \end{align}
   using the fact that $\|\nabla g\|$ is bounded on $\Theta$.
   The next step is to derive an upper bound on $\epsilon$ such that there is no solution in the right branch $[\delta_2, +\infty)$. 
   This way,  the worst-case upper bound on the parameter distance $\|\theta^\mathcal{C} - \theta^\mathcal{C}_\epsilon\| \le L_\theta \epsilon$ is guaranteed to hold.
   To exclude that such points are generated, one can use the compactness of the parameter space $\Theta$. 
   Compactness of $\Theta$ guarantees that there exist $m, r>0$ such that $\dist(\theta, \Theta^\mathcal{C})>r$ implies $\|\nabla \ell(\theta)\|> m$. 
   In combination with \eqref{nablas} we get that $\dist(\theta_\epsilon, \Theta^*)>r$ implies $\epsilon\|\nabla g(\theta)\|\ge m$. 
   This directly implies that, if $\dist(\theta, \Theta^*)>r,$ then $\epsilon \ge \frac{m}{M_{\|\nabla g\|}}=: \epsilon_r$. 
   Therefore, if $\epsilon<\epsilon_r,$ there is no  $\theta_\epsilon \in \Theta^*_\epsilon$ such that $\dist(\theta_\epsilon, \Theta^*)>r$. 
   By imposing $r \le \frac{\mu}{L_\ell^{(1)}}$, we implicitly show that the parameter distance cannot be within the right branch of \eqref{range}.
   Since \eqref{distancebound} holds in the left branch of $\eqref{range}$, this gives us an upper bound on the parameter distance that is linear in $\epsilon$ for all $\epsilon \le \min\{\epsilon_r,\epsilon_\Delta\}$.
   \\

   \textbf{Step 2 - Eigenvalue variation}
   We proceed now by upper bounding the Hessian eigenvalue variation due to the noise. 
   We start using Weyl's inequality, namely
    
   \begin{equation}
      |\lambda_i(\theta^\mathcal{C}) -\tilde{\lambda}_i(\theta^\mathcal{C}_\epsilon)| \leq \| \nabla^2 \ell (\theta^\mathcal{C}) - \nabla^2 \tilde{\ell}(\theta^\mathcal{C}_{\epsilon})\| \label{weyl},
   \end{equation}
   where $\lambda_i(\theta^\mathcal{C}) \in \sigma(\nabla^2 \ell (\theta^\mathcal{C}))$ is the $i$-th eigenvalue of the noise-free Hessian in $\theta^\mathcal{C}$, and 
     $\tilde\lambda_i(\theta^\mathcal{C}_\epsilon) \in \sigma(\nabla^2 \tilde\ell (\theta^\mathcal{C}_\epsilon))$ is the $i$-th eigenvalue of the perturbed Hessian in $\theta^\mathcal{C}_\epsilon$.
   
   Now we upper bound the Hessian distance as
   \begin{align}
      \nonumber & \| \nabla^2 \ell (\theta^\mathcal{C}) - \nabla^2 \tilde{\ell}(\theta^\mathcal{C}_{\epsilon})\| \\
      \nonumber & = \| \nabla^2 \ell (\theta^\mathcal{C}) - \nabla^2 \ell(\theta^\mathcal{C}_{\epsilon}) - \epsilon \nabla^2 g(\theta^\mathcal{C}_{\epsilon})\| \\
      \nonumber & \leq \|\nabla^2 \ell (\theta^\mathcal{C}) - \nabla^2 \ell(\theta^\mathcal{C}_{\epsilon})\| + \epsilon \|\nabla^2 g(\theta^\mathcal{C}_{\epsilon})\| \\
      & \leq L_\mathrm{H} \|\theta^\mathcal{C} - \theta^\mathcal{C}_{\epsilon}\| + \epsilon M_g^{(2)}, \label{partialbound}
   \end{align} 
   with $L_\mathrm{H}:= L_\ell^{(2)}$, since $\nabla^2 \ell$ is Lipschitz-continuous on $\Theta$. A precise derivation of the Lipschitz constant $L_\mathrm{H}$ can be found, for example, in \cite{Zunkovic2026}.
   
   Then by combining \eqref{distancebound}--\eqref{partialbound}we get 
   \begin{align}
      & |\lambda_i(\theta^\mathcal{C}) -\tilde{\lambda}_i(\theta^\mathcal{C}_\epsilon)| \leq L_\mathrm{H} L_\theta \epsilon + M_g^{(2)} \epsilon = (L_\mathrm{H} L_\theta + M_g^{(2)})\epsilon.
   \end{align}

   Next, we show that for small $\epsilon$ negative eigenvalues do not change sign. To do so, observe that
   \begin{align}
      & \tilde\lambda_i(\theta^\mathcal{C}_\epsilon)= \lambda_i(\theta^\mathcal{C}) + \tilde\lambda_i(\theta^\mathcal{C}_\epsilon) - \lambda_i(\theta^\mathcal{C}) 
      \\ & \le \lambda_i(\theta^\mathcal{C}) + |\tilde\lambda_i(\theta^\mathcal{C}) - \lambda_i(\theta^\mathcal{C})| 
      \le \lambda_i(\theta^\mathcal{C}_\epsilon) + (L_\mathrm{H} L_\theta + M_g^{(2)})\epsilon. \label{eigbound}
   \end{align}
   Using that $\lambda(\theta^\mathcal{C}_i)$ is negative, imposing the right-hand side of \eqref{eigbound} to be negative is equivalent to
   \begin{align}
      (L_\mathrm{H} L_\theta + M_g^{(2)})\epsilon < |\lambda(\theta^\mathcal{C})|.  \label{eigcondition}
   \end{align}  

   In particular, in order to guarantee that saddles do not degenerate into local minima, we focus on the negative eigenvalue closest to $0$ for each saddle submanifold. 
   Under the assumption of local surjectivity, \eqref{nominalproblem} has $n_{eig}$ critical sets $\Theta^i$, where $n_{eig}$ is the number of distinct eigenvalues of $O$, of which $n_{eig}-2$ are saddles \cite{Wiedmann2025}. 
   We name the union of the saddle sets $\Theta^\mathcal{S}$. Notice $\Theta^\mathcal{C} \subset \Theta^\mathcal{S}$.
   Further, at each point in the parameter set $\theta \in \Theta$, the Hessian $\nabla^2 \ell (\theta)$ has $p$ real eigenvalues $\lambda \in \sigma(\nabla^2 \ell (\theta))$. 
   We define $\sigma^-(\nabla^2 \ell (\theta)) \subset \sigma(\nabla^2 \ell (\theta))$ as the subset which only contains eigenvalues $\lambda<0$.
   Among these, we find one with the smallest possible distance to $0$ as 
   \begin{equation}
      \bar\lambda = \min_{\substack{\theta^\mathcal{S} \in \Theta^S,\\ \lambda \in \sigma^-(\nabla^2\ell(\theta^\mathcal{S}))}} |\lambda(\theta^\mathcal{S})|.
   \end{equation}
   Note that this minimum is attained as $\lambda$ is a continuous function in $\theta$ and the parameter space $\Theta$ is compact. 
   If \eqref{eigcondition} holds for $\bar\lambda$ then it holds for every other negative eigenvalue $\lambda(\theta^\mathcal{S})\in \sigma^-(\nabla^2\ell(\theta^\mathcal{S}))$ for each $\theta^\mathcal{S} \in \Theta^\mathcal{S}$.
   
   One can see that, if $\epsilon$ satisfies   
      \begin{align}
      \epsilon < {\epsilon}_\lambda=\frac{\bar\lambda}{L_\mathrm{H} L_\theta + M_g^{(2)}},
   \end{align}
   then \eqref{eigcondition} is satisfied for each $\theta^\mathcal{C} \in \Theta^\mathcal{C}$.
   Taking $\epsilon<\bar\epsilon:=\min \{\epsilon_\lambda, \epsilon_\Delta, \epsilon_r\}$ the index is preserved for each critical submanifold.
   Therefore, no local minimum is generated.
   We then proceed with the same proof strategy as for Theorem~1 \cite{Wiedmann2025}. 
   Since all singular points remain either saddles or global extrema and gradient descent avoids saddles, the algorithm either converges to a global minimum $\theta^* \in \Theta^*$ for almost all initial conditions or diverges.
   Finally, divergence is avoided as the parameter space is assumed to be compact.

\end{proof}
Theorem 3 shows that robust convergence of an error perturbed VQE \eqref{perturbedobs} can be guaranteed if the unitary map $U$ is locally surjective and the perturbation is small. 
The proof works in a similar fashion to the noise-free case, by guaranteeing that all critical point that are not global extrema are saddles. We do so by bounding the variation of the Hessian at critical points, accounting both for the variation of the cost function itself and that of the critical point. 
Notice that the proof ensures the existence of an upper bound on the perturbation level $\bar\epsilon$ under potentially conservative assumptions: it only uses worst case estimates and it guarantees that no single eigenvalue of the Hessians can change sign. 
The convergence condition would actually be satisfied as long as at least one of the negative eigenvalues for each saddle remains negative. 
Also, the proof does not rule out that the noise breaks symmetries and splits degenerate submanifolds into lower dimentional separated critical submanifolds. Theorem~3 guarantees that, if this happens, 
these remain saddles as the sign of the other eigenvalues is preserved.  
\\   

Taken together, the Theorems 2 and 3 show that, assuming local surjectivity of the parametrized unitary and for small perturbations on the observable, we can guarantee robust convergence of the perturbed VQE to a neighborhood of the noise-free solution, with radius proportional to the perturbation level itself. 
This result extends asymptotic convergence guarantees for the noise-free problem to a more realistic setup. Although small perturbation levels $\epsilon$ need to be guaranteed, previously there was no evidence that convergence is preserved under the action of any noise process.

\section{Converting Noise into Perturbed Observables}
In the following, we derive a way to represent noise processes acting on the VQE as perturbations on the observable.
Using this, all the results from sections III and IV can be applied to VQEs perturbed by noise. 
Subsection V.A shows how to represent general coherent noise as a perturbed observable, and subsection V.B does the same for incoherent noise.  
\subsection{Coherent noise}
In this subsection, we show how to represent coherent errors as perturbed observables. 
A general coherent-error-perturbed VQE reads
\begin{align}
    \min_{\theta \in \Theta} \Tr[O \tilde{U}(\mathbf{\theta, \eta})\rho_0\tilde{U}^\dagger(\mathbf{\theta, \eta})], \label{coherentperturbedproblem}
\end{align} 
where 
\begin{align}
   \tilde{U}(\mathbf{\theta, \eta})=\prod_{j=1}^L \tilde{U}_j (\theta_j, \eta_j)=\prod_{j=1}^L U_{\rm{e},j}(\eta_j)U_j(\theta_j),
\end{align}
with $U_{\rm{e},j}(\eta_j)=e^{-i\eta_j H_{\mathrm{e},j}}$ for some Hermitian generator $H_{\mathrm{e},j} \in \mathbb{H}^{N \times N}$, $\eta \in \mathbb{R}^p$, 
and $U(\theta_j)$ is the noise-free parameterized gate as described in subsection II.A. 
In other words, this error models the action of an unknown unitary gate after each layer of the VQA.
% We say that a perturbed problem is \textit{associated} to a noise-free problem \eqref{nominalproblem} if it has the same observable $O$ and unitary ansatz $U(\theta)$.

In the following, we show that an error-perturbed VQE can be represented as a perturbed observable problem such that the results from sections III and IV are applicable.
We provide the following definitions to state our result:
\begin{align}
   & V_{k,j}(\theta)=\prod_{m=j}^k U_m (\theta_m) \text{, for } 0<j<k\le L, \text{ and }\\
   & \tilde{V}_{k,j}(\theta, \eta)=\prod_{m=j}^k \tilde{U}_m (\theta_m, \eta_m) \text{, for }  0<j<k\le L.
\end{align}
The circuit segment $V_{k,j}$ models all the gates between the $j$-th and $k$-th layer of the noise-free circuit, while $\tilde{V}_{k,j}$ models the same for the coherent-error-perturbed circuit.

\begin{theorem}
   Problem \eqref{coherentperturbedproblem} is equivalent, up to first-order approximation in $\eta$, to \eqref{perturbedobs}, with $\epsilon=\|\eta\|_\infty$, $\tilde{O}(\theta)=\sum_{j=1}^{L} \frac{\eta_j}{\|\eta\|_\infty} [\hat{H}_{\mathrm{e},j}(\theta), O]$ and $\hat{H}_{\mathrm{e},j}(\theta)=V_{L,j}(\theta)H_{\mathrm{e},j}V_{L,j}^\dagger(\theta)$.
\end{theorem}
We say that two optimization problems are equivalent if they have the same cost $\ell(\theta)$ up to scaling or offset.

Theorem 4 will be proven via two intermediate Lemmas below. 
We start by focusing on a problem with the structure 
\begin{align}
    \min_{\theta \in \Theta} \Tr[O U_\mathrm{e}(\eta, \theta){U}(\mathbf{\theta})\rho_0{U}^\dagger(\mathbf{\theta})U^\dagger_\mathrm{e}(\eta, \theta)], \label{coherentperturbedlast}
\end{align} 
with $U_\mathrm{e}(\eta, \theta)=\prod_{j=1}^{L}U_{\mathrm{e},j}(\eta_j, \theta)$, and $U_{\mathrm{e},j}(\eta_j, \theta)=e^{-i \eta_j \hat{H}_{\mathrm{e},j}(\theta)}$. 
Note that here the coherent errors depend on the algorithm parameters $\theta$.
First, we show that problem \eqref{coherentperturbedlast} can be approximated by a perturbed observable problem \eqref{perturbedobs} for a suitably defined observable $\tilde{O}(\theta)$, and then we prove that any VQE affected by coherent errors 
can be represented in this fashion.

\begin{lemma}
   Problem \eqref{coherentperturbedlast} is equivalent to \eqref{perturbedobs} up to first order approximation in $\eta$, with $\epsilon=\|\eta\|_\infty$ and
   $\tilde{O}(\theta)=\sum_{j=1}^{L} \frac{\eta_j}{\|\eta\|_\infty} [\hat{H}_{\mathrm{e},j}(\theta), O]$.  
\end{lemma}
\begin{proof}

   The first step is to show that \eqref{coherentperturbedlast} is equivalent to \eqref{nominalproblem} with a rotated observable.
   Indeed, 
   \begin{align}
   & \Tr[O U_\mathrm{e}(\eta, \theta){U}(\mathbf{\theta})\rho_0{U}^\dagger(\mathbf{\theta})U^\dagger_\mathrm{e}(\eta, \theta)] \nonumber\\
   &=\Tr[U^\dagger_\mathrm{e}(\eta, \theta) O U_\mathrm{e}(\eta, \theta){U}(\mathbf{\theta})\rho_0{U}^\dagger(\mathbf{\theta})] \nonumber \\
   &=\Tr[\hat{O}(\eta, \theta)U(\theta)\rho_0{U}^\dagger(\mathbf{\theta})], 
\end{align} 
defining $\hat{O}(\eta, \theta)=U^\dagger_\mathrm{e}(\eta, \theta) O U_\mathrm{e}(\eta, \theta)$.
Taking a first-order approximation of $\hat{O}(\eta, \theta)$ in $\eta$ we get
\begin{align}
   &\hat{O}(\eta, \theta)= O + \sum_{j=1}^{L}\eta_j [\hat{H}_{\mathrm{e},j}(\theta),O]+\mathcal{O}(\|\eta\|^2) \nonumber \\
   &= O + \|\eta\|_\infty \sum_{j=1}^{L}\frac{\eta_j}{\|\eta\|_\infty}[\hat{H}_{\mathrm{e},j}(\theta),O]+\mathcal{O}(\|\eta\|^2), 
   % &= O + \epsilon \sum_{j=1}^{L}\lambda_j [H_{\mathrm{e},j},O]+\mathcal{O}(\|\eta\|^2),
\end{align}
and we define $\epsilon=\|\eta\|_\infty$.
This proves our statement.
\end{proof}

Next, we show that any VQE affected by coherent errors can be rewritten in the form \eqref{coherentperturbedlast}.
We note that a strategy similar to Lemma 2 was pursued in \cite{Berberich2025} for the robustness analysis of generic quantum circuits.
It remains to prove that \eqref{coherentperturbedproblem} can be rewritten as \eqref{coherentperturbedlast}. 
We can do this by exploiting the interaction picture. 
\begin{lemma}
   Problem \eqref{coherentperturbedproblem} is equivalent to \eqref{coherentperturbedlast} with 
   $\hat{H}_{\mathrm{e},j}(\theta)=V_{L,j}(\theta)H_{\mathrm{e},j}V_{L,j}^\dagger(\theta)$.
\end{lemma}

\begin{proof}
   To prove the lemma, we show how noise on one intermediate layer can be pushed onto the observable. The statement follows by iteratively repeating the procedure
   over the circuit. 
   It is possible to show that a circuit $\tilde{U}(\theta, \eta)$ with a perturbation acting up to the $j$-th layer can be represented as
   \begin{align}
      & \tilde{U}(\theta, \eta) =V_{L,j}(\theta)U_{\mathrm{e,j}}(\eta_j)\tilde{V}_{j,1}(\theta) \nonumber \\
      & = \tilde{U}_{\mathrm{e,j}}(\eta_j, \theta) V_{L,j}(\theta)\tilde{V}_{j,1}(\theta), 
   \end{align} 
   with $\tilde{U}_{\mathrm{e,j}}(\eta_j, \theta)  =V^\dagger_{L,j}(\theta) U_{\mathrm{e,j}}(\eta_j) V_{L,j}(\theta)$.
   By Taylor expansion of the matrix exponential, it is easy to check that 
   \begin{align}
      V^\dagger_{L,j}(\theta) \tilde{U}_{\mathrm{e,j}}(\eta_j, \theta) V_{L,j}(\theta)=e^{-i \eta_j V^\dagger_{L,j}(\theta) H_{\mathrm{e,j}}V_{L,j}(\theta)}.   
   \end{align}
   Defining $\tilde{H}_{\mathrm{e},j}=V^\dagger_{L,j}(\theta) H_{\mathrm{e,j}}V_{L,j}(\theta)$,
   \begin{align}
      \tilde{U}(\theta, \eta) =e^{-i \eta_j \tilde{H}_{\mathrm{e},j}} U(\theta).
   \end{align}
   In case noise is present on each layer, we repeat the procedure for each layer iteratively and our statement is proven.
\end{proof}
It remains to show that $\tilde{O}(\theta)$ is smooth in $\theta$. This is easily done by checking that it is a finite sum of smooth elements.
It follows that the hypotheses on $\tilde{O}(\theta)$ for Theorem 2 hold for VQEs affected by general coherent errors.

\subsection{Incoherent noise}
In this subsection, we show how to represent incoherent errors as perturbed observables. 
We consider incoherent error-perturbed VQEs structured as
\begin{align}
    \min_{\theta \in \Theta} \Tr[O \mathcal{E}_\mathbf{\theta}(\rho_0)], \label{perturbedproblem}
\end{align} 
where 
\begin{align}
   &\mathcal{E}_\mathbf{\theta}(\rho_0)=\mathcal{E}_{L,\mathbf{\theta}_L} \circ \mathcal{E}_{L-1,\mathbf{\theta}_{L-1}} \circ \dots \circ \mathcal{E}_{1,\mathbf{\theta}_1} (\rho_0)\text{,} \\
   &\mathcal{E}_{j,\mathbf{\theta}_j}(\rho)=\mathcal{E}_j (U_j(\mathbf{\theta_j}) \rho U_j(\mathbf{\theta_j})^\dagger)\text{, and } \\
   &\mathcal{E}_j(\rho)=(1-p)\rho + p \sum_{k=1}^{m} w_{j,k} E_{j,k} \rho E_{j,k}^\dagger,
\end{align}
with Kraus operators $E_{j,k}$ with norm $\|E_{j,k}\| \le 1$ for each $k,j$, and $\sum_{k=1}^{m} w_{j,k}=1$.
This quantum operation models the action of incoherent noise after each layer. The only restriction on the noise is that it must leave the input density operator $\rho$ unaltered with probability $1-p$.
This includes, for example, all Pauli noise channels. Physically, this can be interpreted as a channel such that there is at least one situation in which the environment does not gain information about the state of the system. 
Amplitude damping channels, for example, do not satisfy this property, as emissions are impossible from relaxed states, therefore rendering the lack of an emission event itself informative. 
We proceed with the same approach as in subsection V.A. 

We begin by stating our main result: 
\begin{theorem}
   Problem \eqref{perturbedproblem} is equivalent, up to a constant scaling factor of the cost, to \eqref{perturbedobs}, with $\epsilon=\frac{p}{1-p}$,
   $\tilde{O}(\theta)=\sum_{k=1}^{m} w_{k} \tilde E_{k}^\dagger(\theta) O \tilde E_{k}(\theta)$ and suitably defined $\tilde{E}_k(\theta)$.
\end{theorem}

We again prove the result via two intermediate lemmas.

To this end, we first consider the perturbed VQE 
\begin{align}
   \min_{\theta \in \Theta} \Tr[O \mathcal{E}_{L,\theta}(U(\mathbf{\theta}) \rho_0 U(\mathbf{\theta})^\dagger)], \label{perturbedlast}
\end{align}
with $\mathcal{E_{L,\theta}(\rho)}=(1-p)\rho + p \sum_{k=1}^{m} w_k \tilde E_k(\theta) \rho \tilde E_k^\dagger(\theta)$, $\tilde E_k(\theta)$ being a Kraus operator with norm $\|\tilde E_k\| \le 1$, possibly depending on the parameters, and $\sum_{k=1}^{m} w_k=1$.
First we show that quantum channels on the output of the circuit can be interpreted as perturbations on the observable. 
Then, we show how quantum channels can be represented as acting on the output as per \eqref{perturbedlast}.

\begin{lemma}
   Problem \eqref{perturbedlast} is equivalent, up to a scaling factor, to \eqref{perturbedobs} with $\epsilon=\frac{p}{1-p}$ and
   $\tilde{O}(\theta)=\sum_{k=1}^{m} w_k \tilde E_k^\dagger(\theta) O \tilde E_k(\theta)$.
\end{lemma}

\begin{proof}
The proof follows from the following steps
\begin{align}
    & \tilde{\ell} (\theta) = \Tr[O \mathcal{E}_{L,\theta}(\rho(\theta))] \nonumber \\
    & = \Tr[O ((1-p)\rho(\theta) + p \sum_{k=1}^{m} w_k \tilde E_k(\theta) \rho(\theta) \tilde E_k^\dagger(\theta))] \nonumber \\
    & = \Tr[((1-p)O + p \sum_{k=1}^{m} w_k \tilde E_k^\dagger(\theta)  O \tilde E_k(\theta))\rho(\theta)] \nonumber \\
    & = (1-p) \Tr[(O + \epsilon \tilde{O}(\theta))\rho(\theta)],
\end{align}
where we define $\epsilon=\frac{p}{1-p}$ as the perturbation level,
$\tilde{O}(\theta)=\sum_{k=1}^{m} w_k \tilde E_k^\dagger(\theta)  O \tilde E_k(\theta)$. 
Not that the scaling $1-p$ does not affect the optimal decision variable.
\end{proof}
% the case where noise is modeled only on the last layer of the circuit. This is the case, for example, for a single-layer $SU(d)$ circuit. 
Lemma 3 shows that it is possible to represent a parameter-dependent quantum operation acting on the output of the quantum circuit $U(\theta)$ as a perturbed observable. 
In the following, we show that \eqref{perturbedlast} can model any incoherent noise process described by \eqref{perturbedproblem}.

We can state the following lemma: 
\begin{lemma}
   Problem \eqref{perturbedproblem} is equivalent to \eqref{perturbedlast} with 
   $p=1-\prod_j (1-p_j)$ and suitably defined operators $\tilde{E}_{j,k}(\theta)$. 
\end{lemma}

\begin{proof}
The proof can be decomposed in two steps.
First, we show that noise acting on any intermediate layer can be seen as a different quantum operation acting on the 
output layer.
Second, we show how to combine nested quantum operations into a single one. 
Combining this two results and using them iteratively for each layer proves the lemma. 

In order to show that mid circuit noises can be represented as a different noise process on the output,
we start defining 
\begin{align}
   \rho_j(\theta)=\mathcal{E}_{j-1,\mathbf{\theta}_{j-1}} \circ \mathcal{E}_{j-2,\mathbf{\theta}_{j-2}} \circ \dots \circ \mathcal{E}_{1,\mathbf{\theta}_1} (\rho_0), 
\end{align}
as the (perturbed) density operator at the $j$-th layer of the circuit. 
Using this, we can define a VQE affected by incoherent noise up to the $j$-th layer as 
\begin{align}
    \min_{\theta \in \Theta} \tilde{\ell} (\theta) = \Tr[O V_{L,j}(\theta) \mathcal{E}_j(\rho_j(\theta))V_{L,j}^\dagger (\theta)],
\end{align}
we can derive the following
\begin{align}
    & \Tr[O V_{L,j}(\theta) \mathcal{E}_j(\rho_j(\theta))V_{L,j}^\dagger (\theta)] \nonumber \\
    & = \Tr[O V_{L,j}(\theta) ((1-p)\rho_j(\theta) + \nonumber p \sum_{k=1}^{m} w_{j,k} E_{j,k} \rho_j(\theta) E^\dagger_{j,k}) V_{L,j}^\dagger (\theta)] \nonumber \\
    & = \Tr[O ((1-p)V_{L,j}(\theta)\rho_j(\theta) V_{L,j}^\dagger (\theta) +  \nonumber \\
    & \quad \quad \quad \quad \quad \quad p \sum_{k=1}^{m} w_{j,k} V_{L,j}(\theta)E_{j,k} \rho_j(\theta) E^\dagger_{j,k}V_{L,j}^\dagger (\theta))]  \nonumber \\
    & = \Tr[O ((1-p)V_{L,j}(\theta)\rho_j(\theta) V_{L,j}^\dagger (\theta) + \nonumber \\
    & \quad \quad \quad \quad \quad \quad p \sum_{k=1}^{m} w_k \tilde{E}_{j,k}(\theta)V_{L,j}(\theta)\rho_j(\theta) V_{L,j}^\dagger (\theta)\tilde{E}^\dagger_{j,k}(\theta))] \nonumber \\
    & = \Tr[O \tilde{\mathcal{E}_\theta}(V_{L,j}(\theta) \rho_j(\theta) V_{L,j}^\dagger (\theta))]    
\end{align}
with $\tilde{E}_{j,k}(\theta) = V_{L,j}(\theta)E_{j,k} V_{L,j}^\dagger (\theta)$. 

For the second step, we use the well-known fact that the composition of quantum operations is itself a quantum operation. 
To derive the error probability $p=1-\prod_j (1-p_j)$ it is sufficient to observe that the probability that the input density operator $\rho$ 
remains unaltered in all the quantum operations is $q=\prod_j (1-p_j)$. It follows immediately that the error probability $p=1-q$.
The Kraus operators can be obtained as all the possible choices of Kraus operators.
By combining the two results and iterating for each layer we get a structure that matches \eqref{perturbedlast}.
\end{proof}

Lemma 4 can also be used to characterize how the perturbation level changes with increasing circuit depth. Indeed, 
assuming constant error probability $p$ for each layer, one gets $\epsilon(L)=\frac{1-(1-p)^L}{(1-p)^L}=(1-p)^{-L}-1$.
This suggests an exponential scaling of the perturbation level with respect to the circuit depth, in line with the well-known exponential decay of information due to quantum channels, shown for example in \cite{Nielsen2012}.

%%%%%%%%%%%%%%%%%%%%%%%%%%%%%%%%%%%%%%%%%%%%%%%%%%%%%%%%%
\section{Special cases}
In this subsection, we discuss specific cases for which we can derive stronger results. 

\subsection{Coherent Control errors}
Coherent control errors are multiplicative errors that model imprecisions in the application of the control pulses. 
This can happen, for example, in case of miscalibration. 
A parameterized unitary $U_j(\theta_j)$ affected by a coherent control error is modeled as a unitary 
\begin{align}
   \tilde{U}_j(\theta_j, \eta_j)=e^{-i\theta H_j}e^{-i\eta_j\theta_j H_j}=e^{-i(1+\eta_j)\theta_j H_j}.
\end{align}
Since, for this class of errors, the noisy optimal parameters are just a scaling of the noise-free ones, it is easy to show that 
\begin{align}
   \tilde{\theta}^*_{\eta,j}=\frac{1}{1+\eta_j}\theta^*_j, 
\end{align}
where by $\tilde{\theta}^*_{\eta,j}$ we indicate the $j$-th component of the noisy VQE optimal parameter $\tilde{\theta}^*_{\eta}$. 
It follows that  
\begin{align}
   \|\tilde{\theta}^*_\eta - \theta^* \| \le \|\eta\|_\infty \|\tilde{\theta}^*_\eta\|. 
\end{align}
This is a tight bound, in the sense that it is always saturated for at least one component of $\theta^*_\eta$, that holds globally in the parameter space $\Theta$.
Additionally, from this bound one can see that L2 regularization of the parameter vector $\theta$ as proposed in \cite{Berberich2024b} improves perturbation bounds on optimal parameters as well. 
\\

\subsection{Depolarization Noise} A statement stronger than the bound given by Theorem 1 can be made for the specific case of depolarization noise.
Depolarization noise is defined as 
\begin{align}
   \mathcal{E}_{DN}(\rho)=(1-p)\rho + {p}I. \label{DN}
\end{align} 
In this case we can state that the optimal solution set of the noise-perturbed VQE is exactly the same as the noise-free one. 
\begin{lemma}
   Let an incoherent noise-perturbed VQE \eqref{perturbedlast} be affected by depolarization noise \eqref{DN}.
   Then $\Theta^*_\epsilon=\Theta^*$ for each $\epsilon>0$.
\end{lemma}

\begin{proof}
   The statement can be proven directly as
   \begin{align}
   & \Theta^*_\epsilon=\argmin_{\theta \in \Theta} \tilde{\ell} (\theta) = \argmin\Tr[O \mathcal{E}(U(\mathbf{\theta}) \rho_0 U(\mathbf{\theta})^\dagger)] \nonumber \\ 
   &=\argmin  (1-p) \Tr[O U(\mathbf{\theta})\rho_0 U^\dagger(\mathbf{\theta})] + p Tr[O] \nonumber \\
   &= \argmin  (1-p) \Tr[O U(\mathbf{\theta})\rho_0 U^\dagger(\mathbf{\theta})] = \Theta^* 
\end{align}
\end{proof}
The result highlights that, since depolarization corresponds essentially to a uniform shrinking of the Bloch sphere, it does not affect the correct rotation for the state regardless of the intensity. 
Notice that, although it does not affect asymptotic convergence, depolarization noise represents a major obstacle for VQA training
because it has been shown to induce barren plateaus \cite{Wang2021}. 

\subsection{Noise on output layer} 
Convergence properties from the noise-free VQE transfer to the noisy ones, 
if the perturbation observable is constant over $\theta$.
For perturbed observable VQEs \eqref{perturbedobs}, we can state the following corollary.
\begin{corollary}
   Let $U$ be locally surjective and $\Theta$ be compact. 
   Assume that $\tilde{O}(\theta)=\tilde{O}$ be constant. Then, for almost all initial conditions $\theta_0$ \eqref{perturbedobs} either diverges
   or converges to a global minimum.
\end{corollary}

\begin{proof}
   The partial derivatives of \eqref{perturbedobs} read
\begin{align}
   & \dfrac{\partial}{\partial \theta_j} \ell(\theta)=\braket{\grad \tilde{L} [U (\mathbf{\theta})], \Omega_j (\mathbf{\theta})} + \Tr[\rho(\theta)\dfrac{\partial }{\partial \theta_j}\tilde{O}(\theta)] \nonumber \\
   & = \braket{\grad \tilde{L} [U (\mathbf{\theta})], \Omega_j (\mathbf{\theta})},
\end{align}
where 
\begin{align}
   & \grad \tilde{L}[U(\theta)]=[O+ \epsilon \tilde{O}, \rho(\mathbf{\theta})] \text{,}   
\end{align}
and $\Omega_j(\theta)$ is defined as in the noise-free case and we used the fact that $\dfrac{\partial }{\partial \theta_j}\tilde{O}(\theta) = 0 $ for each $\theta$. 
The result follows applying \cite[Theorem 1]{Wiedmann2025} and using compactness of the parameter space. 
\end{proof}
Corollary 1 shows that perturbations on the observable that are not dependent on the circuit parameters do not affect convergence properties of the VQE. 
This includes every noise model that has a perturbation only right before measurement, including coherent and decoherent noise on a $SU(N)$ Ansatz and 
measurement noise for misaligned measurements.
%%%%%%%%%%%%%%%%%%%%%%%%%%%%%%%%%%%%%%%%%%%%%%%%%%%%%%%%%%%%%%%%%%%%%%%%%%%%%%%%%%%%

\section{Numerics}
In this section, we provide numerical validation for our results. All simulations are developed using Pennylane \cite{Bergholm2018}.
The results are obtained training different VQEs on a 5-qubit simulated quantum circuit. All circuits are trained for $1000$ iterations.
The three circuits we simulate are a VQE for a randomly generated Hamiltonian, QAOA for a designed graph and a VQC circuit for the quantum fourier transform. 
We tuned the step size for each circuit individually. For both the randomly generated VQE and the VQC we used the locally surjective ansatz proposed in \cite{Wiedmann2025}.
The VQC and random VQE have depth $L=2$ and QAOA has depth $L=30$.

\begin{figure}[h]
   \includegraphics[width=\columnwidth]{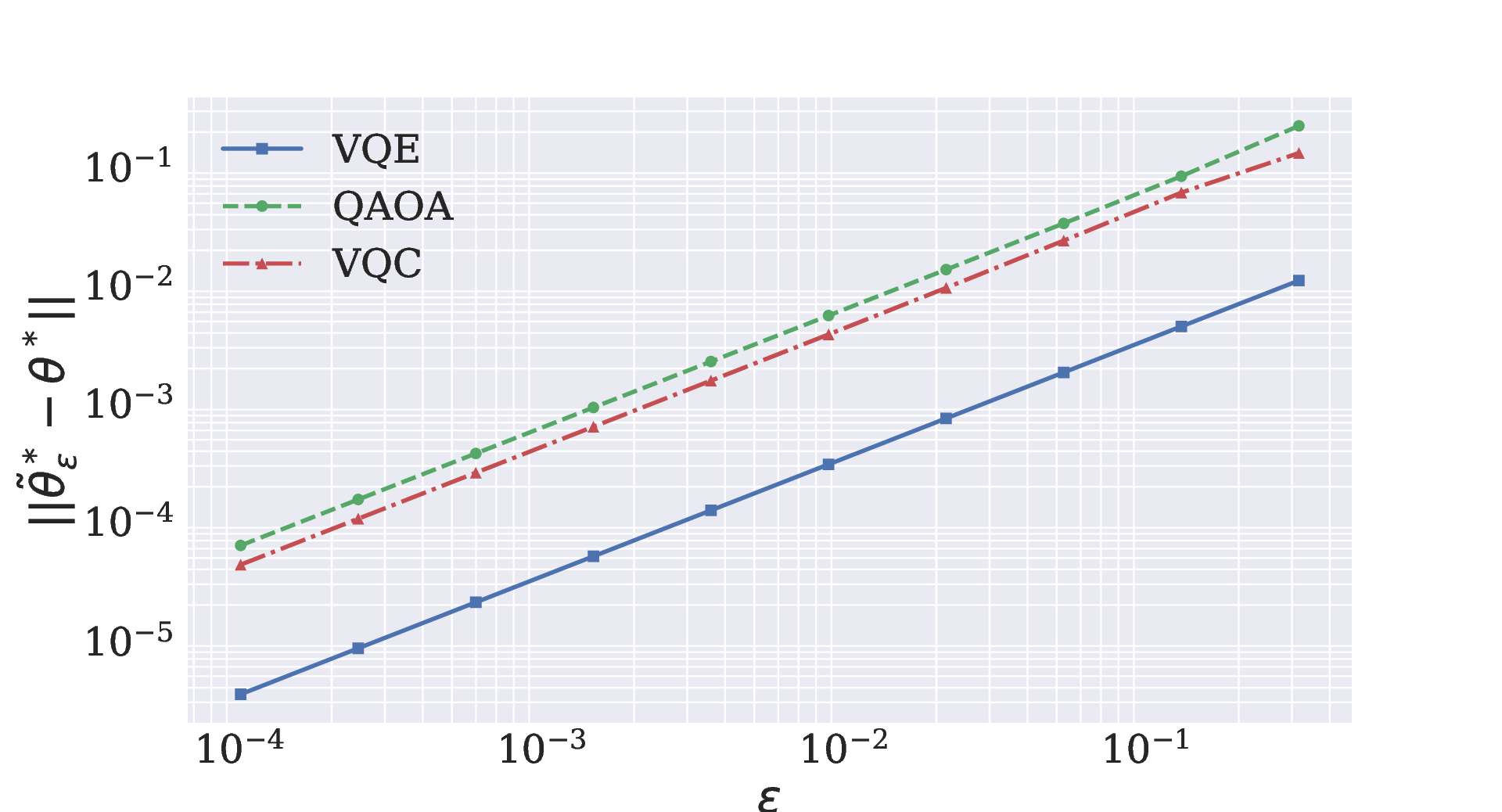}   
   \caption{Distance of the perturbed optimal parameter from the noise-free one over the perturbation level for coherent errors}
   \label{coherentplot}
\end{figure}
  
\begin{figure}[h]
   \includegraphics[width=\columnwidth]{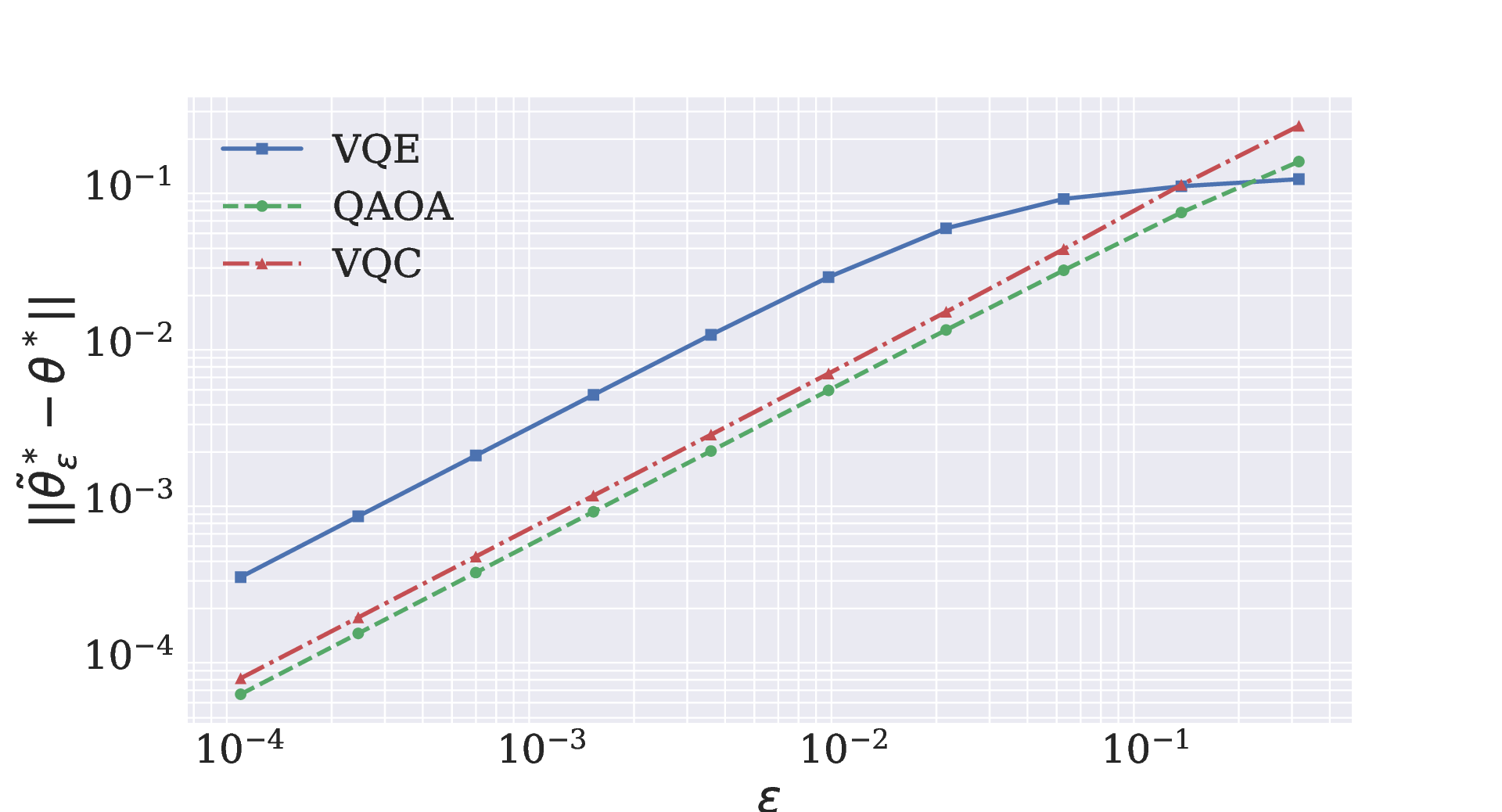}   
   \caption{Distance of the perturbed optimal parameter from the noise-free one over the perturbation level for incoherent errors}
   \label{incoherentplot}
\end{figure}

We first simulated coherent errors. We did so by adding a Z rotation gate with constant parameter $\eta=\epsilon$ after each gate.
We solved the VQE for perturbation levels ranging from $\epsilon=10^{-4}$ to $\epsilon=10^{-\frac{1}{2}}$.

Figure \ref{coherentplot} shows the results of our simulations. In accordance to Theorem 1, we found that in all three cases the distance increases linearly with the perturbation level. 
Further, in our simulations, noisy algorithm initialized with the same initial parameters as their noiseless counterpart converge to a value that is close to the noise-free solution, not only to the 
optimal subset. 

Figure \ref{incoherentplot} shows the same behavior for incoherent errors. We modeled incoherent errors with a bit flip quantum operation with probability $p=1-(\frac{1}{1+\epsilon})^\frac{1}{L}$. 
This guarantees we generate errors that lead to the desired perturbation level. 
%%%%%%%%%%%%%%%%%%%%%%%%%%%%%%%%%%%%%%%%%%%%%%%%%%%%%%%%%%%%%%%%%%%%%%%%%%%%%%%%
\section{Conclusions and future works}

\subsection{Conclusions}
In this work, we investigated the robustness of optimal parameters for VQEs with respect to quantum noise. 
We proved that, if the observable is perturbed, the perturbation on the optimal parameter scales proportionally to the 
perturbation level, as long as the perturbation level is small.
We then showed that this formulation includes both general coherent errors and a wide class of incoherent errors.
We proved that, as long as the perturbation observable does not depend on the VQE parameters, the same convergence guarantees as the
noise-free VQE can be recovered. Further, we have provided numerical evidence that supports our results.

\subsection{Future Works}
There are many possible extensions to this work. While Theorem 1 assumed a compact parameter space, studying the non-compact case is practically relevant
and numerical evidence supports the possibility of generalizing our results to this case. 
This work could also be extended to other classes of VQAs, potentially including to quantum machine learning applications, where the interest is mainly in the optimal parameters rather than the optimal cost.   
Last, the derived analysis results could be used to design a synthesis procedure to obtain inherently robust VQEs. 
%%%%%%%%%%%%%%%%%%%%%%%%%%%%%%%%%%%%%%%%%%%%%%%%%%%%%%%%%%%%%%%%%%%%%%%%%%%%%%%%

%%%%%%%%%%%%%%%%%%%%%%%%%%%%%%%%%%%%%%%%%%%%%%%%%%%%%%%%%%%%%%%%%%%%%%%%%%%%%%%%

\section*{Acknowledgments} 

This work was funded by QuantERA FeedbacQ, by the Deutsche Forschungsgemeinschaft (DFG, German
Research Foundation) – 579821331 and 583407438, and co-funded by the European Commission.

\bibliographystyle{IEEEtran}
\bibliography{bibliography}

@Article{Preskill2018,
  author     = {Preskill, John},
  journal    = {Quantum},
  title      = {Quantum Computing in the {NISQ} era and beyond},
  year       = {2018},
  issn       = {2521-327X},
  month      = aug,
  pages      = {79},
  volume     = {2},
  doi        = {10.22331/q-2018-08-06-79},
  publisher  = {Verein zur Forderung des Open Access Publizierens in den Quantenwissenschaften},
  readstatus = {skimmed},
}

@Article{Shor1997,
  author    = {Shor, Peter W.},
  journal   = {SIAM Journal on Computing},
  title     = {Polynomial-Time Algorithms for Prime Factorization and Discrete Logarithms on a Quantum Computer},
  year      = {1997},
  issn      = {1095-7111},
  month     = oct,
  number    = {5},
  pages     = {1484--1509},
  volume    = {26},
  doi       = {10.1137/s0097539795293172},
  publisher = {Society for Industrial & Applied Mathematics (SIAM)},
}

@book{Nesterov2004,
  author    = {Yurii Nesterov},
  title     = {Introductory Lectures on Convex Optimization: A Basic Course},
  series    = {Applied Optimization},
  volume    = {87},
  publisher = {Springer},
  address   = {New York, NY},
  year      = {2004},
  edition   = {1},
  doi       = {10.1007/978-1-4419-8853-9},
  isbn      = {978-1-4020-7553-7},
}

@misc{Zunkovic2026,
      title={Scalable, self-verifying variational quantum eigensolver using adiabatic warm starts}, 
      author={Bojan Žunkovič and Marco Ballarin and Lewis Wright and Michael Lubasch},
      year={2026},
      eprint={2602.17612},
      archivePrefix={arXiv},
      primaryClass={quant-ph},
      url={https://arxiv.org/abs/2602.17612}, 
}

@Article{Cerezo2021,
  author     = {Cerezo, M. and Arrasmith, Andrew and Babbush, Ryan and Benjamin, Simon C. and Endo, Suguru and Fujii, Keisuke and McClean, Jarrod R. and Mitarai, Kosuke and Yuan, Xiao and Cincio, Lukasz and Coles, Patrick J.},
  journal    = {Nature Reviews Physics},
  title      = {Variational quantum algorithms},
  year       = {2021},
  issn       = {2522-5820},
  month      = aug,
  number     = {9},
  pages      = {625--644},
  volume     = {3},
  doi        = {10.1038/s42254-021-00348-9},
  publisher  = {Springer Science and Business Media LLC},
  readstatus = {read},
}

@misc{cheng2013,
      title={Quantum Geometric Tensor (Fubini-Study Metric) in Simple Quantum System: A pedagogical Introduction}, 
      author={Ran Cheng},
      year={2013},
      eprint={1012.1337},
      archivePrefix={arXiv},
      primaryClass={quant-ph},
      url={https://arxiv.org/abs/1012.1337}, 
}

@unpublished{simonetto2025,
  author       = {Simonetto, Andrea},
  title        = {Differentiable Optimisation: Theory and Algorithms -- Part {II}: Algorithms},
  note         = {Lecture notes, ENSTA Paris. HAL: hal-04416966v2},
  institution  = {ENSTA Paris},
  year         = {2025},
  type         = {Engineering school},
  url          = {https://hal.science/hal-04416966v2}
}

@Article{Peruzzo2014,
  author     = {Peruzzo, Alberto and McClean, Jarrod and Shadbolt, Peter and Yung, Man-Hong and Zhou, Xiao-Qi and Love, Peter J. and Aspuru-Guzik, Alán and O’Brien, Jeremy L.},
  journal    = {Nature Communications},
  title      = {A variational eigenvalue solver on a photonic quantum processor},
  year       = {2014},
  issn       = {2041-1723},
  month      = jul,
  number     = {1},
  volume     = {5},
  doi        = {10.1038/ncomms5213},
  publisher  = {Springer Science and Business Media LLC},
  readstatus = {skimmed},
}

@Article{McClean2016,
  author     = {McClean, Jarrod R and Romero, Jonathan and Babbush, Ryan and Aspuru-Guzik, Alán},
  journal    = {New Journal of Physics},
  title      = {The theory of variational hybrid quantum-classical algorithms},
  year       = {2016},
  issn       = {1367-2630},
  month      = feb,
  number     = {2},
  pages      = {023023},
  volume     = {18},
  doi        = {10.1088/1367-2630/18/2/023023},
  publisher  = {IOP Publishing},
  readstatus = {skimmed},
}

@Article{Sharma2020,
  author     = {Sharma, Kunal and Khatri, Sumeet and Cerezo, M and Coles, Patrick J},
  journal    = {New Journal of Physics},
  title      = {Noise resilience of variational quantum compiling},
  year       = {2020},
  issn       = {1367-2630},
  month      = apr,
  number     = {4},
  pages      = {043006},
  volume     = {22},
  doi        = {10.1088/1367-2630/ab784c},
  publisher  = {IOP Publishing},
  readstatus = {read},
}

@Article{Bittel2021,
  author    = {Bittel, Lennart and Kliesch, Martin},
  journal   = {Physical Review Letters},
  title     = {Training Variational Quantum Algorithms Is NP-Hard},
  year      = {2021},
  issn      = {1079-7114},
  month     = sep,
  number    = {12},
  pages     = {120502},
  volume    = {127},
  doi       = {10.1103/physrevlett.127.120502},
  publisher = {American Physical Society (APS)},
}

@Article{Tilly2022,
  author    = {Tilly, Jules and Chen, Hongxiang and Cao, Shuxiang and Picozzi, Dario and Setia, Kanav and Li, Ying and Grant, Edward and Wossnig, Leonard and Rungger, Ivan and Booth, George H. and Tennyson, Jonathan},
  journal   = {Physics Reports},
  title     = {The Variational Quantum Eigensolver: A review of methods and best practices},
  year      = {2022},
  issn      = {0370-1573},
  month     = nov,
  pages     = {1--128},
  volume    = {986},
  doi       = {10.1016/j.physrep.2022.08.003},
  publisher = {Elsevier BV},
}

@Misc{Farhi2014,
  author     = {Farhi, Edward and Goldstone, Jeffrey and Gutmann, Sam},
  title      = {A Quantum Approximate Optimization Algorithm},
  year       = {2014},
  copyright  = {arXiv.org perpetual, non-exclusive license},
  doi        = {10.48550/ARXIV.1411.4028},
  publisher  = {arXiv},
  readstatus = {skimmed},
}

@Article{BravoPrieto2023,
  author    = {Bravo-Prieto, Carlos and LaRose, Ryan and Cerezo, M. and Subasi, Yigit and Cincio, Lukasz and Coles, Patrick J.},
  journal   = {Quantum},
  title     = {Variational Quantum Linear Solver},
  year      = {2023},
  issn      = {2521-327X},
  month     = nov,
  pages     = {1188},
  volume    = {7},
  doi       = {10.22331/q-2023-11-22-1188},
  publisher = {Verein zur Forderung des Open Access Publizierens in den Quantenwissenschaften},
}

@Article{Keijzer2022,
  author    = {de Keijzer, R. J. P. T. and Colussi, V. E. and Škorić, B. and Kokkelmans, S. J. J. M. F.},
  journal   = {AVS Quantum Science},
  title     = {Optimization of the variational quantum eigensolver for quantum chemistry applications},
  year      = {2022},
  issn      = {2639-0213},
  month     = feb,
  number    = {1},
  volume    = {4},
  doi       = {10.1116/5.0076435},
  publisher = {American Vacuum Society},
}

@Article{Grimsley2019,
  author    = {Grimsley, Harper R. and Economou, Sophia E. and Barnes, Edwin and Mayhall, Nicholas J.},
  journal   = {Nature Communications},
  title     = {An adaptive variational algorithm for exact molecular simulations on a quantum computer},
  year      = {2019},
  issn      = {2041-1723},
  month     = jul,
  number    = {1},
  volume    = {10},
  doi       = {10.1038/s41467-019-10988-2},
  publisher = {Springer Science and Business Media LLC},
}

@Misc{Wiedmann2025,
  author     = {Wiedmann, Marco and Burgarth, Daniel and Dirr, Gunther and Schulte-Herbrüggen, Thomas and Malvetti, Emanuel and Arenz, Christian},
  title      = {On the convergence of the variational quantum eigensolver and quantum optimal control},
  year       = {2025},
  copyright  = {Creative Commons Attribution 4.0 International},
  doi        = {10.48550/ARXIV.2509.05295},
  publisher  = {arXiv},
  readstatus = {read},
}

@Article{Wiersema2024,
  author     = {Wiersema, Roeland and Lewis, Dylan and Wierichs, David and Carrasquilla, Juan and Killoran, Nathan},
  journal    = {Quantum},
  title      = {Here comes the SU(N): multivariate quantum gates and gradients},
  year       = {2024},
  issn       = {2521-327X},
  month      = mar,
  pages      = {1275},
  volume     = {8},
  doi        = {10.22331/q-2024-03-07-1275},
  publisher  = {Verein zur Forderung des Open Access Publizierens in den Quantenwissenschaften},
  readstatus = {read},
}

@Article{Khatri2019,
  author    = {Khatri, Sumeet and LaRose, Ryan and Poremba, Alexander and Cincio, Lukasz and Sornborger, Andrew T. and Coles, Patrick J.},
  journal   = {Quantum},
  title     = {Quantum-assisted quantum compiling},
  year      = {2019},
  issn      = {2521-327X},
  month     = may,
  pages     = {140},
  volume    = {3},
  doi       = {10.22331/q-2019-05-13-140},
  publisher = {Verein zur Forderung des Open Access Publizierens in den Quantenwissenschaften},
}

@Article{Magann2021,
  author     = {Magann, Alicia B. and Arenz, Christian and Grace, Matthew D. and Ho, Tak-San and Kosut, Robert L. and McClean, Jarrod R. and Rabitz, Herschel A. and Sarovar, Mohan},
  journal    = {PRX Quantum},
  title      = {From Pulses to Circuits and Back Again: A Quantum Optimal Control Perspective on Variational Quantum Algorithms},
  year       = {2021},
  issn       = {2691-3399},
  month      = jan,
  number     = {1},
  pages      = {010101},
  volume     = {2},
  doi        = {10.1103/prxquantum.2.010101},
  publisher  = {American Physical Society (APS)},
  readstatus = {read},
}

@Article{Chakrabarti2007,
  author     = {Chakrabarti, Raj and Rabitz, Herschel},
  journal    = {International Reviews in Physical Chemistry},
  title      = {Quantum control landscapes},
  year       = {2007},
  issn       = {1366-591X},
  month      = oct,
  number     = {4},
  pages      = {671--735},
  volume     = {26},
  doi        = {10.1080/01442350701633300},
  publisher  = {Informa UK Limited},
  readstatus = {skimmed},
}

@Article{Ge2022,
  author     = {Ge, Xiaozhen and Wu, Re-Bing and Rabitz, Herschel},
  journal    = {Annual Reviews in Control},
  title      = {The optimization landscape of hybrid quantum-classical algorithms: From quantum control to {NISQ} applications},
  year       = {2022},
  issn       = {1367-5788},
  pages      = {314--323},
  volume     = {54},
  doi        = {10.1016/j.arcontrol.2022.06.001},
  publisher  = {Elsevier BV},
  readstatus = {read},
}

@Article{Larocca2023,
  author     = {Larocca, Martín and Ju, Nathan and García-Martín, Diego and Coles, Patrick J. and Cerezo, Marco},
  journal    = {Nature Computational Science},
  title      = {Theory of overparametrization in quantum neural networks},
  year       = {2023},
  issn       = {2662-8457},
  month      = jun,
  number     = {6},
  pages      = {542--551},
  volume     = {3},
  doi        = {10.1038/s43588-023-00467-6},
  publisher  = {Springer Science and Business Media LLC},
  readstatus = {read},
}

@Article{Kurdyka2014,
  author     = {Kurdyka, Krzysztof and Spodzieja, Stanisław},
  journal    = {Proceedings of the American Mathematical Society},
  title      = {Separation of real algebraic sets and the Łojasiewicz exponent},
  year       = {2014},
  issn       = {1088-6826},
  month      = may,
  number     = {9},
  pages      = {3089--3102},
  volume     = {142},
  doi        = {10.1090/s0002-9939-2014-12061-2},
  publisher  = {American Mathematical Society (AMS)},
  readstatus = {skimmed},
}

@Article{Feehan2020,
  author    = {Feehan, Paul M. N.},
  journal   = {Calculus of Variations and Partial Differential Equations},
  title     = {On the Morse-Bott property of analytic functions on Banach spaces with Łojasiewicz exponent one half},
  year      = {2020},
  issn      = {1432-0835},
  month     = apr,
  number    = {2},
  volume    = {59},
  doi       = {10.1007/s00526-020-01734-4},
  publisher = {Springer Science and Business Media LLC},
}

@Article{Kim2023,
  author     = {Kim, Youngseok and Eddins, Andrew and Anand, Sajant and Wei, Ken Xuan and van den Berg, Ewout and Rosenblatt, Sami and Nayfeh, Hasan and Wu, Yantao and Zaletel, Michael and Temme, Kristan and Kandala, Abhinav},
  journal    = {Nature},
  title      = {Evidence for the utility of quantum computing before fault tolerance},
  year       = {2023},
  issn       = {1476-4687},
  month      = jun,
  number     = {7965},
  pages      = {500--505},
  volume     = {618},
  doi        = {10.1038/s41586-023-06096-3},
  publisher  = {Springer Science and Business Media LLC},
  readstatus = {skimmed},
}

@Article{Arute2019,
  author    = {Arute, Frank and Arya, Kunal and Babbush, Ryan and Bacon, Dave and Bardin, Joseph C. and Barends, Rami and Biswas, Rupak and Boixo, Sergio and Brandao, Fernando G. S. L. and Buell, David A. and Burkett, Brian and Chen, Yu and Chen, Zijun and Chiaro, Ben and Collins, Roberto and Courtney, William and Dunsworth, Andrew and Farhi, Edward and Foxen, Brooks and Fowler, Austin and Gidney, Craig and Giustina, Marissa and Graff, Rob and Guerin, Keith and Habegger, Steve and Harrigan, Matthew P. and Hartmann, Michael J. and Ho, Alan and Hoffmann, Markus and Huang, Trent and Humble, Travis S. and Isakov, Sergei V. and Jeffrey, Evan and Jiang, Zhang and Kafri, Dvir and Kechedzhi, Kostyantyn and Kelly, Julian and Klimov, Paul V. and Knysh, Sergey and Korotkov, Alexander and Kostritsa, Fedor and Landhuis, David and Lindmark, Mike and Lucero, Erik and Lyakh, Dmitry and Mandrà, Salvatore and McClean, Jarrod R. and McEwen, Matthew and Megrant, Anthony and Mi, Xiao and Michielsen, Kristel and Mohseni, Masoud and Mutus, Josh and Naaman, Ofer and Neeley, Matthew and Neill, Charles and Niu, Murphy Yuezhen and Ostby, Eric and Petukhov, Andre and Platt, John C. and Quintana, Chris and Rieffel, Eleanor G. and Roushan, Pedram and Rubin, Nicholas C. and Sank, Daniel and Satzinger, Kevin J. and Smelyanskiy, Vadim and Sung, Kevin J. and Trevithick, Matthew D. and Vainsencher, Amit and Villalonga, Benjamin and White, Theodore and Yao, Z. Jamie and Yeh, Ping and Zalcman, Adam and Neven, Hartmut and Martinis, John M.},
  journal   = {Nature},
  title     = {Quantum supremacy using a programmable superconducting processor},
  year      = {2019},
  issn      = {1476-4687},
  month     = oct,
  number    = {7779},
  pages     = {505--510},
  volume    = {574},
  doi       = {10.1038/s41586-019-1666-5},
  publisher = {Springer Science and Business Media LLC},
}

@Article{Berberich2024,
  author     = {Berberich, Julian and Fink, Daniel},
  journal    = {IEEE Control Systems},
  title      = {Quantum Computing Through the Lens of Control: A Tutorial Introduction},
  year       = {2024},
  issn       = {1941-000X},
  month      = dec,
  number     = {6},
  pages      = {24--49},
  volume     = {44},
  doi        = {10.1109/mcs.2024.3466448},
  publisher  = {Institute of Electrical and Electronics Engineers (IEEE)},
  readstatus = {read},
}

@Article{Huembeli2021,
  author    = {Huembeli, Patrick and Dauphin, Alexandre},
  journal   = {Quantum Science and Technology},
  title     = {Characterizing the loss landscape of variational quantum circuits},
  year      = {2021},
  issn      = {2058-9565},
  month     = feb,
  number    = {2},
  pages     = {025011},
  volume    = {6},
  doi       = {10.1088/2058-9565/abdbc9},
  publisher = {IOP Publishing},
}

@Article{Sweke2020,
  author     = {Sweke, Ryan and Wilde, Frederik and Meyer, Johannes and Schuld, Maria and Faehrmann, Paul K. and Meynard-Piganeau, Barthélémy and Eisert, Jens},
  journal    = {Quantum},
  title      = {Stochastic gradient descent for hybrid quantum-classical optimization},
  year       = {2020},
  issn       = {2521-327X},
  month      = aug,
  pages      = {314},
  volume     = {4},
  doi        = {10.22331/q-2020-08-31-314},
  publisher  = {Verein zur Forderung des Open Access Publizierens in den Quantenwissenschaften},
  readstatus = {read},
}

@Article{Harrow2021,
  author    = {Harrow, Aram W. and Napp, John C.},
  journal   = {Physical Review Letters},
  title     = {Low-Depth Gradient Measurements Can Improve Convergence in Variational Hybrid Quantum-Classical Algorithms},
  year      = {2021},
  issn      = {1079-7114},
  month     = apr,
  number    = {14},
  pages     = {140502},
  volume    = {126},
  doi       = {10.1103/physrevlett.126.140502},
  publisher = {American Physical Society (APS)},
}

@Article{Berberich2024a,
  author    = {Berberich, Julian and Fink, Daniel and Pranjić, Daniel and Tutschku, Christian and Holm, Christian},
  journal   = {Physical Review Research},
  title     = {Training robust and generalizable quantum models},
  year      = {2024},
  issn      = {2643-1564},
  month     = dec,
  number    = {4},
  pages     = {043326},
  volume    = {6},
  doi       = {10.1103/physrevresearch.6.043326},
  publisher = {American Physical Society (APS)},
}

@Article{Kandala2017,
  author    = {Kandala, Abhinav and Mezzacapo, Antonio and Temme, Kristan and Takita, Maika and Brink, Markus and Chow, Jerry M. and Gambetta, Jay M.},
  journal   = {Nature},
  title     = {Hardware-efficient variational quantum eigensolver for small molecules and quantum magnets},
  year      = {2017},
  issn      = {1476-4687},
  month     = sep,
  number    = {7671},
  pages     = {242--246},
  volume    = {549},
  doi       = {10.1038/nature23879},
  publisher = {Springer Science and Business Media LLC},
}

@Article{Wiersema2020,
  author    = {Wiersema, Roeland and Zhou, Cunlu and de Sereville, Yvette and Carrasquilla, Juan Felipe and Kim, Yong Baek and Yuen, Henry},
  journal   = {PRX Quantum},
  title     = {Exploring Entanglement and Optimization within the Hamiltonian Variational Ansatz},
  year      = {2020},
  issn      = {2691-3399},
  month     = dec,
  number    = {2},
  pages     = {020319},
  volume    = {1},
  doi       = {10.1103/prxquantum.1.020319},
  publisher = {American Physical Society (APS)},
}

@Article{Wang2021,
  author     = {Wang, Samson and Fontana, Enrico and Cerezo, M. and Sharma, Kunal and Sone, Akira and Cincio, Lukasz and Coles, Patrick J.},
  journal    = {Nature Communications},
  title      = {Noise-induced barren plateaus in variational quantum algorithms},
  year       = {2021},
  issn       = {2041-1723},
  month      = nov,
  number     = {1},
  volume     = {12},
  doi        = {10.1038/s41467-021-27045-6},
  publisher  = {Springer Science and Business Media LLC},
  readstatus = {skimmed},
}

@Article{Berberich2024b,
  author     = {Berberich, Julian and Fink, Daniel and Holm, Christian},
  journal    = {Physical Review A},
  title      = {Robustness of quantum algorithms against coherent control errors},
  year       = {2024},
  issn       = {2469-9934},
  month      = jan,
  number     = {1},
  pages      = {012417},
  volume     = {109},
  doi        = {10.1103/physreva.109.012417},
  publisher  = {American Physical Society (APS)},
  readstatus = {read},
}

@Misc{Berberich2025,
  author     = {Berberich, Julian and Fellner, Tobias and Kosut, Robert L. and Holm, Christian},
  title      = {Robustness of quantum algorithms: Worst-case fidelity bounds and implications for design},
  year       = {2025},
  copyright  = {arXiv.org perpetual, non-exclusive license},
  doi        = {10.48550/ARXIV.2509.08481},
  publisher  = {arXiv},
  readstatus = {read},
}

@Misc{Bergholm2018,
  author    = {Bergholm, Ville and Izaac, Josh and Schuld, Maria and Gogolin, Christian and Ahmed, Shahnawaz and Ajith, Vishnu and Alam, M. Sohaib and Alonso-Linaje, Guillermo and AkashNarayanan, B. and Asadi, Ali and Arrazola, Juan Miguel and Azad, Utkarsh and Banning, Sam and Blank, Carsten and Bromley, Thomas R and Cordier, Benjamin A. and Ceroni, Jack and Delgado, Alain and Di Matteo, Olivia and Dusko, Amintor and Garg, Tanya and Guala, Diego and Hayes, Anthony and Hill, Ryan and Ijaz, Aroosa and Isacsson, Theodor and Ittah, David and Jahangiri, Soran and Jain, Prateek and Jiang, Edward and Khandelwal, Ankit and Kottmann, Korbinian and Lang, Robert A. and Lee, Christina and Loke, Thomas and Lowe, Angus and McKiernan, Keri and Meyer, Johannes Jakob and Montañez-Barrera, J. A. and Moyard, Romain and Niu, Zeyue and O'Riordan, Lee James and Oud, Steven and Panigrahi, Ashish and Park, Chae-Yeun and Polatajko, Daniel and Quesada, Nicolás and Roberts, Chase and Sá, Nahum and Schoch, Isidor and Shi, Borun and Shu, Shuli and Sim, Sukin and Singh, Arshpreet and Strandberg, Ingrid and Soni, Jay and Száva, Antal and Thabet, Slimane and Vargas-Hernández, Rodrigo A. and Vincent, Trevor and Vitucci, Nicola and Weber, Maurice and Wierichs, David and Wiersema, Roeland and Willmann, Moritz and Wong, Vincent and Zhang, Shaoming and Killoran, Nathan},
  title     = {PennyLane: Automatic differentiation of hybrid quantum-classical computations},
  year      = {2018},
  copyright = {arXiv.org perpetual, non-exclusive license},
  doi       = {10.48550/ARXIV.1811.04968},
  publisher = {arXiv},
}

@Book{Nielsen2012,
  author     = {Nielsen, Michael A. and Chuang, Isaac L.},
  publisher  = {Cambridge University Press},
  title      = {Quantum Computation and Quantum Information: 10th Anniversary Edition},
  year       = {2012},
  isbn       = {9780511976667},
  month      = jun,
  doi        = {10.1017/cbo9780511976667},
  readstatus = {read},
}

@Book{Nicolaescu2011,
  author    = {Nicolaescu, Liviu},
  publisher = {Springer New York},
  title     = {An Invitation to Morse Theory},
  year      = {2011},
  isbn      = {9781461411055},
  doi       = {10.1007/978-1-4614-1105-5},
  issn      = {2191-6675},
  journal   = {Universitext},
}

@Misc{Colding2014,
  author     = {Colding, Tobias Holck and Minicozzi, William P.},
  title      = {Lojasiewicz inequalities and applications},
  year       = {2014},
  copyright  = {arXiv.org perpetual, non-exclusive license},
  doi        = {10.48550/ARXIV.1402.5087},
  publisher  = {arXiv},
  readstatus = {skimmed},
}

\end{document}